\documentclass[12pt,onecolumn,letterpaper]{IEEEtran}
\usepackage{color}
\usepackage{dsfont}
\usepackage{multirow}
\usepackage{setspace}
\usepackage{algorithmic}
\usepackage{graphicx}
\usepackage{epsfig}
\usepackage{epstopdf}
\usepackage{colortbl}
\usepackage{theoremref}
\usepackage{amssymb}
\newtheorem{theorem}{Theorem}
\newtheorem{lemma}{Lemma}
\newtheorem{remark}{Remark}
\definecolor{myblue}{gray}{0.9}

\newtheorem{definition}{Definition}

\newtheorem{example}{Example}

\begin{document}

\title{An Efficient Algorithm for Finding Dominant Trapping Sets of LDPC Codes${}^{\ast}$ \thanks{${}^{\ast}$A preliminary
version of part of this work appeared in Proc. {\em 6th International Symposium on Turbo Codes \& Iterative Information Processing,} Brest, France, Sept. 6 - 10, 2010.} }

\author{Mehdi~Karimi,~\IEEEmembership{Student Member,~IEEE}
        and Amir~H.~Banihashemi,~\IEEEmembership{Senior Member,~IEEE}}
\maketitle

\thispagestyle{empty}

 \linespread{1.6}
 \selectfont

\begin{abstract}

This paper presents an efficient algorithm for finding the dominant trapping sets of a low-density
parity-check (LDPC) code. The algorithm can be used to estimate the error floor of
LDPC codes or to be used as a tool to design LDPC codes with low error floors.
For regular codes, the algorithm is initiated with a set of short cycles as the input. For irregular codes, in addition to short cycles, variable nodes with low degree and cycles with low approximate cycle extrinsic message degree (ACE) are also used as the initial inputs. The initial inputs are then expanded recursively to dominant trapping sets of increasing size.
At the core of the algorithm lies the analysis of the graphical structure
of dominant trapping sets and the relationship of such structures to short cycles, low-degree variable nodes and cycles with low ACE.
The algorithm is universal in the sense that it can be used for an arbitrary graph and
that it can be tailored to find a variety of graphical objects, such as absorbing sets and Zyablov-Pinsker (ZP)
trapping sets, known to dominate the performance of LDPC codes in the error floor region over
different channels and for different iterative decoding algorithms.
Simulation results on several LDPC codes demonstrate the
accuracy and efficiency of the proposed algorithm. In particular,
the algorithm is significantly faster than the existing search algorithms
for dominant trapping sets.
\end{abstract}


\section{Introduction}

\IEEEPARstart{E}{stimating} the error floor performance of low-density parity-check (LDPC) codes
under iterative message-passing decoding, and the design of LDPC codes with low error floors have
attracted a great amount of interest in recent years.
The performance of LDPC codes under iterative decoding
algorithms in the error floor region is closely related to the structure of the code's
Tanner graph. For the binary erasure channel (BEC), the problematic structures are
\textit{stopping sets }\cite{Di2002}. In the case of the binary
symmetric channel (BSC) and the additive white Gaussian noise (AWGN) channel,
the error-prone patterns are called \textit{trapping sets} \cite{Richardson2003},
\textit{near codewords} \cite{Mackay2002} or \textit{pseudo codewords} \cite{Vontobel2005}.
Among the trapping sets, the so-called \textit{elementary trapping sets} are
shown to be the main
culprits~\cite{Richardson2003}, \cite{Laendner2005}, \cite{Cole2008}, \cite{Milenkovict2007}, \cite{Laendner2009},
\cite{Zhang2009}.
Related to this, it is demonstrated in~\cite{Dolecek2007} that for some structured LDPC
codes decoded by iterative algorithms over the AWGN channel, a subset of trapping sets,
called \textit{absorbing sets}, determine the error floor performance.
In fact, in an overwhelmingly large number of cases, dominant absorbing sets
appear to be elementary trapping sets.

For a given LDPC code, the knowledge of dominant trapping sets is important. On one hand, efficient methods for estimating the error floor of an LDPC code, which rely on the importance sampling technique, operate by biasing the noise toward the dominant trapping sets of the code, see, e.g., \cite{Cole2008}. On the other hand, by knowing the dominant trapping sets, several decoder modifications can be applied to improve the error floor performance (see, e.g., \cite{Cavus2005}, \cite{Han2008}). Furthermore, the knowledge of dominant trapping sets can be used to design LDPC codes with low error floor.  Related to this, Ivkovic \emph{et al.} \cite{Ivkovic2008} applied the technique of
edge swapping between two copies of a base LDPC code to eliminate the dominant trapping sets of the base code over the BSC. This was then generalized by Asvadi \emph{et al.} \cite{Asvadi2010} to cyclic liftings of higher degree to construct quasi-cyclic LDPC codes with low error floor. While the knowledge of the problematic sets that dominate the
error floor performance is most helpful in the design and analysis of LDPC codes, attaining such knowledge, regardless of differences in the graphical structure of these sets, is a hard problem. For instance, it was shown in ~\cite{Murali2006}, \cite{Krishnan2007}, \cite{McGregor2008} that the problem of finding a minimum size stopping set is NP hard. McGregor and Milenkovic \cite{McGregor2008} showed that not only the problem of finding a  minimum size trapping set, but also the problem of approximating the size of a minimal trapping set is NP hard, regardless of the sparsity of the underlying graph.

It should be noted that while the majority of the literature on estimating the error floor of LDPC codes rely on finding the dominant trapping sets, as an eventual result of decoder failure, in \cite{XB-07}, \cite{Xiao2009}, \cite{XB-08}, Xiao and Banihashemi took a different approach. By focussing on the input error patterns that cause the decoder to fail, they developed a simple technique to estimate the frame error rate (FER) and the bit error rate (BER) of finite-length LDPC codes over the BSC \cite{XB-07}. The complexity of this technique was then reduced in \cite{Xiao2009}, and the estimation technique was extended to the AWGN channel with
quantized output in \cite{XB-08}. In this work, unlike \cite{XB-07}, \cite{Xiao2009}, \cite{XB-08}, we are particularly interested in the examination of the graphical structure of the problematic sets which dominate the error floor performance. This information is then used to efficiently search for these sets.

The complexity of the exhaustive brute force search method for
finding problematic structures of size $t$ in a code of length $n$ becomes quickly infeasible as $n$ and $t$ increase.
Efficient search algorithms have been devised to find
small (dominant) stopping and trapping sets~\cite{Wang2007}, \cite{Cole2008}, \cite{Rosnes2009}, \cite{Xiao2009}, \cite{Wang2009}, \cite{Abu2010}, \cite{Kyung2010}.
The reach of these algorithms however is still very limited. For example,
the complexity of the algorithm of~\cite{Wang2007}, \cite{Wang2009} is only affordable for codes
with lengths up to $\sim 500$. Even for these lengths, the algorithm can
only find trapping sets of maximum size $11$ with only one or two
unsatisfied check nodes. This is while for many codes, some of the dominant trapping sets
may have larger size and/or more than two unsatisfied check nodes. In \cite{Vasic2009} and \cite{Nguyen2011}, the authors proposed to build a database of all possible configurations for trapping sets of different sizes in a graph with specific degree distribution and girth. They then used a parent-child relationship between the trapping sets of different sizes to simplify the search of the larger trapping sets. This method was used to find the dominant trapping sets of left regular LDPC codes with left degree $3$ \cite{Vasic2009}. The method proposed in \cite{Vasic2009} however becomes very complex when the degree of variable nodes, and in turn the number of possible configurations increases. The application of this method becomes even more difficult when dealing with irregular LDPC codes as for such codes, there may be a large number of possible configurations for each type of trapping set, due to the variety of variable node degrees. Even for the left regular graphs with small left degrees, the number of possible configurations becomes quite large for the larger trapping sets. It is therefore important
to look for more efficient algorithms to find the problematic structures
that dominate the error floor performance of LDPC codes.

In this paper, we study the problematic graphical structures that dominate the error floor
performance of LDPC codes, collectively referred to as trapping sets, and demonstrate that they
all contain at least one short cycle (with a small exception of some of the trapping sets of irregular LDPC codes with degree-2 variable nodes). By examining the relationships between cycles and
trapping sets, we devise an efficient algorithm to find dominant trapping sets
of an LDPC code. The algorithm is initiated by a set of short cycles as input.
Each cycle is then expanded recursively to trapping sets of increasing size in a conservative fashion,
i.e., the expanded sets all have the smallest size larger than
the size of the current set, and each of them will be used as a new input to
the next step of the algorithm. It should be mentioned that although our algorithm uses the topological relationships between the small trapping sets and the larger ones, it is different from the method of \cite{Vasic2009} in several ways. Our algorithm is not based on the knowledge of the exact structure of trapping sets, and hence does not need to build a database. In fact, instead of checking all the possible configurations to find the existing ones in a graph, it directly and efficiently finds those existing configurations, and so it is much faster. Moreover, unlike the method of \cite{Vasic2009}, our algorithm uses a general framework for all the degree distributions and girths (with a small exception of some of the trapping sets of irregular LDPC codes with degree-2 variable nodes).
The proposed algorithm is applicable to any Tanner graph
(structured or random) and can be tailored to find a variety of graphical
structures, such as elementary trapping sets and absorbing sets among others. For structured graphs, such as those of quasi-cyclic or protograph codes, one can use the existing automorphisms in the graph to further simplify the search.
Results on several LDPC codes verify the high efficiency and accuracy of the algorithm.
For example, for the tested codes, the search speed is improved by a factor
of 10 to 100 compared to the methods of \cite{Abu2010} and \cite{Cole2008}.

The remainder of this paper is organized as follows. Basic definitions and notations are provided in Section II. In Section III, we develop the proposed algorithm. Section IV presents the modification needed for irregular LDPC codes. Section V includes some numerical results. Section VI concludes the paper.

\section{Definitions and Notations}

Let $G=(L\cup R\,,E)$ be the bipartite graph, or Tanner graph, corresponding to the
LDPC code $\mathcal{C}$,
where $L$ is the set of variable nodes, $R$ is the set of check nodes and $E$ is the set of edges. The notations $L$ and $R$ refer to ``left" and  ``right", respectively, pointing to the side of the bipartite graph where variable nodes and check nodes are located, respectively.
The degree of a node $v \in L\,\, (or\,\, R)$ is
denoted by $d(v)\,$. For a subset $\mathcal{S} \subset L\,$, $\Gamma(\mathcal{S})$ denotes the set of neighbors of $\mathcal{S} $ in $R\,$.
The \textit{induced subgraph} of $\mathcal{S}$, represented by $G(\mathcal{S})$, is the graph containing nodes $\mathcal{S}\cup \Gamma(\mathcal{S})$ with
edges $\{ (u,v) \in E: u \in \mathcal{S} ,\, v\in \Gamma (\mathcal{S})  \}$. The set of check nodes in $\Gamma(\mathcal{S})$ with odd
degree in $G(\mathcal{S})$ is denoted by $\Gamma_{\mathrm{o}}(\mathcal{S})$. Similarly, $\Gamma_{\mathrm{e}}(\mathcal{S})$ represents the set of check nodes in $\Gamma(\mathcal{S})$ with even degree in $G(\mathcal{S})$.
The subgraph resulting from removing the nodes of $\Gamma_{\mathrm{o}}(\mathcal{S})$ and their edges from $G(\mathcal{S})$ is denoted by $G'(\mathcal{S})$.
In this paper, we interchangeably use the terms \textit{satisfied check nodes} and
\textit{unsatisfied check nodes} to denote the check nodes in $\Gamma_{\mathrm{e}}(\mathcal{S})$ and $\Gamma_{\mathrm{o}}(\mathcal{S})$, respectively.
Given a Tanner graph $G=(L\cup R\,,E)$, the following objects play an important
role in the error floor performance of the corresponding LDPC code:
\begin{definition}
\begin{itemize}
\item []
\item [i)] A set $\mathcal{S} \subset L$ is an $(a,b)$ {\em trapping set} if $|\mathcal{S}|=a$ and $|\Gamma_{\mathrm{o}}(\mathcal{S})|=b$. The integer $a$ is referred to as the {\em size} of the trapping set $\mathcal{S}$.
    \label{def1}
\item [ii)] An $(a,b)$ trapping set $\mathcal{S}$ is called {\em elementary} if all the check nodes in $G(\mathcal{S})$ have degree one or two.\label{def2}
\item [iii)] A set $\mathcal{S} \subset L$ is an $(a,b)$ \emph{absorbing set} if $\mathcal{S}$ is an $(a,b)$ trapping set
and if all the nodes in $\mathcal{S}$ are connected to more nodes in $\Gamma_{\mathrm{e}}(\mathcal{S})$ than to nodes in $\Gamma_{\mathrm{o}}(\mathcal{S})$.
\label{def4}
\item [iv)] A set $\mathcal{S} \subset L$ is an $(a,b)$ \emph{fully absorbing set} if $\mathcal{S}$ is an $(a,b)$ absorbing set and if all the nodes in $L\backslash \mathcal{S}$ have strictly more neighbors in $R\backslash \Gamma_{\mathrm{o}}(\mathcal{S})$ than in $\Gamma_{\mathrm{o}}(\mathcal{S})$.
\item [v)] A set $\mathcal{S} \subset L$ is a \emph{$k$-out trapping} set~\cite{Wang2007} if $\Gamma_{\mathrm{o}}(\mathcal{S})$ contains exactly
$k$ nodes of degree one in $G(\mathcal{S})$.
\label{def5}
\item [vi)] Let $G=(L\cup R\,,E)$ be a left-regular bipartite graph with left degree $l$.
A set $\mathcal{S} \subset L$ is a \textit{Zyablov-Pinsker (ZP) trapping set}~\cite{McGregor2008}
if every node of $\mathcal{S}$ is connected to less than $l-\lfloor (l-1)/2\rfloor$ nodes in $\Gamma_{\mathrm{o}}(\mathcal{S})$.
\label{def6}
\end{itemize}
\end{definition}


The ZP trapping sets are the trapping sets of the Zyablov-Pinsker bit-flipping algorithm~\cite{ZP-76}
over the BSC~\cite{McGregor2008}. It should also be noted that for odd values of
$l$, the definitions of ZP trapping sets and absorbing sets are identical.

It is important to note that Definitions~\ref{def2}(ii) -- \ref{def6}(vi) are all special cases of
a trapping set in Definition~\ref{def1}(i). In the rest of the paper, therefore, we
collectively refer to them as trapping sets. Distinctions will be made as necessary. Trapping sets with smaller values of $a$ and $b$ are generally believed to be more harmful to iterative decoding. Loosely speaking, such trapping sets are called {\em dominant}. To measure how harmful a trapping set really is, one can use techniques such as importance sampling \cite{Cole2008} to measure the contribution of the trapping set to the error floor. This contribution and the dominance of a trapping set (compared to others) would also depend on the channel model and the iterative decoding algorithm, as well as the detailed structure of the Tanner graph (not just the values of $a$ and $b$).

In a graph $G=(V,E)$ with the set of nodes $V$ and the set of edges $E$, a \emph{lollipop walk} of
length $k$ is defined as a sequence of nodes $v_1, v_2, .\,.\,.\,, v_{k+1}$ in $V$
such that $v_1, v_2, .\,.\,.\,, v_{k}$ are distinct, $v_{k+1}=v_{m}$ for some $m \in [2,k]$, and
$(v_i, v_{i+1}) \in E$ for all $i \in \{1,\,.\,.\,.,k\}$.
Fig.~\ref{fig1} shows two lollipop walks of length $7$.
The lollipop walk in Fig.~\ref{fig1}$(a)$ is represented as  $v_1 v_2 v_3 v_4 v_5 v_6 v_7 v_2$.
A {\em cycle} can be considered as a special lollipop walk if the definition is extended
to $m=1$. The length of the shortest cycle in a graph $G$ is denoted by $g$ and
is called the {\em girth} of $G$.

\begin{figure}[!t]
\centering
\includegraphics[width=2.8in]{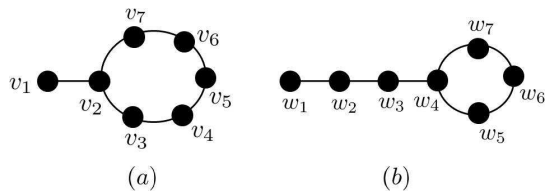}
\caption{Two lollipop walks of length $7$.}
\label{fig1}
\end{figure}

\section{Development of the Proposed Algorithm}

\subsection{Graphical Structure of Trapping Sets}
Without loss of generality, we assume that the induced subgraph of a trapping set is connected.
Disconnected trapping sets can be considered as the union of connected ones. Moreover, to the best of our knowledge, almost all the structures
reported as dominant trapping sets (of regular LDPC codes) in the literature have the property that every variable node is
connected to at least two satisfied check nodes in the induced subgraph. We thus focus on trapping sets with this property except for irregular LDPC codes, where we relax this condition for degree-2 variable nodes. As an example, the subgraph in
Fig.~\ref{fig2} does not satisfy this condition. (Variable nodes and check nodes are represented
by circles and squares, respectively.) Removal of node $v_1$
and its edges however makes the subgraph satisfy the condition. The following lemma proves this property for certain absorbing and ZP trapping sets.
\begin{lemma}
Suppose that $\mathcal{S} \subset L$ is an absorbing set (ZP trapping set) in $G=(L\cup R,E)$,
and that for all variable nodes $v \in \mathcal{S}$, we have $d(v)\ge 2$ ($d(v)\ge 3$). Then each variable node
$v \in \mathcal{S}$ is connected to at least two satisfied check nodes in $G(\mathcal{S})$.
\label{lem1}
\end{lemma}
\begin{proof}
The proof follows from the definition of absorbing and ZP trapping sets.
\end{proof}

For small trapping sets, which dominate the error floor performance, it is unlikely to see check nodes of degree larger than $2$ in their
subgraphs, i.e., most of the dominant trapping sets are elementary~\cite{Cole2008},
\cite{Richardson2003}. Related to this, almost all the trapping sets reported as the dominant trapping sets of practical LDPC codes are elementary. In fact, it can be shown that the sizes of non-elementary trapping sets for left-regular graphs are generally larger than those of elementary ones (cf.~Lemma \ref{lem_nonel} in Appendix A).
 \begin{example}
 For left-regular LDPC codes with left degree 4 and girths 6, 8 and 10, lower bounds on the sizes of non-elementary trapping sets $\mathcal{S}$ with less than 3 unsatisfied check nodes and with at least one satisfied check node of degree larger than 2 in $G(\mathcal{S})$, are 7, 14 and 22, respectively (based on Lemma \ref{lem_nonel} in Appendix A, and by choosing $b = 2$). Moreover, for the same conditions, the minimum sizes of non-elementary trapping sets $\mathcal{S}$ with at least one unsatisfied check node of degree larger than 1 in $G(\mathcal{S})$ (and without satisfied check nodes of degree larger than 2 in $G(\mathcal{S})$) are at least 5, 11 and 17, respectively. This is while for the same scenario, the code can have elementary trapping sets of size 5, 8 and 17, respectively.
 \end {example}

 In the following, we develop our search algorithm mainly for elementary trapping sets, and then present simple modifications to tailor the algorithm to find non-elementary trapping sets.

\begin{figure}[!t]
\centering
\includegraphics[width=1.6in]{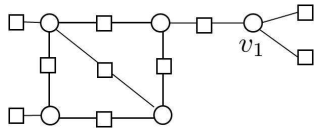}
\caption{The induced subgraph of a trapping set.}
\label{fig2}
\end{figure}

In the rest of the paper, we use the notation $\mathcal{T}$ to denote the set of all trapping sets $\mathcal{S} $ in a graph
$G$ whose induced subgraph $G(\mathcal{S})$ is connected and for which every node $v \in \mathcal{S}$ is connected to at least
two nodes in $\Gamma_{\mathrm{e}}(\mathcal{S})$. Notation $\mathcal{T}^a$ is used for the set of
all elements in $\mathcal{T}$ with size $a$ and $\mathcal{T}_{\mathcal{S}}$ denotes the set of all
elements in $\mathcal{T}$ that contain the set $\mathcal{S}$. Naturally, $\mathcal{T}_{\mathcal{S}}^a$ denotes
the set of all elements in $\mathcal{T}$ of size $a$ that contain the set $\mathcal{S}$. In
the following, we also assume that
the Tanner graph $G$ has no parallel edges and no node of degree less than 2.

\subsection{Expansion of Elementary Trapping Sets}
The main idea of the proposed algorithm is to start from a relatively small set of
small elementary trapping sets, which are easy to enumerate, and then recursively expand them
to larger elementary trapping sets. To achieve this, we first characterize the expansion
of an elementary trapping set to a larger elementary trapping set through the following lemmas.

\begin{lemma}
Let $\mathcal{S}$ be an elementary trapping set of size $a$ in $\mathcal{T}$. Then for each elementary trapping
set $\mathcal{S}' \in \mathcal{T}_{\mathcal{S}}^{a+1}$ (if any), the variable node in $\mathcal{S}'\backslash \mathcal{S}$ is only connected
to unsatisfied check nodes of $\mathcal{S}$ (i.e., to the check nodes in $\Gamma_{\mathrm{o}}(\mathcal{S})$).
\label{lem2}
\end{lemma}

\begin{proof}
If the node in $\mathcal{S}'\backslash \mathcal{S}$ is connected to any satisfied check nodes of $\mathcal{S}$,
then $\mathcal{S}'$ will have unsatisfied check nodes in $\Gamma_{\mathrm{o}}(\mathcal{S}')$ connected to
$3$ variable nodes of $\mathcal{S}'$. This contradicts $\mathcal{S}'$ being an elementary trapping set.
\end{proof}
Fig.~\ref{fig3}$(a)$ depicts an example of the set $\mathcal{S}'$ discussed in Lemma~\ref{lem2}. (It should be noted that in all the configurations of Fig.~\ref{fig3}, including \ref{fig3}$(a)$, unsatisfied check nodes of $G({\mathcal{S}')}$ are not shown.)

\begin{figure}[!t]
\centering
\includegraphics[width=2.6in]{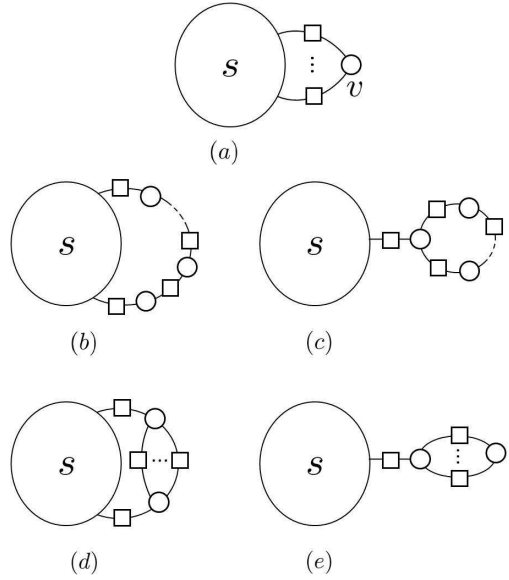}
\caption{Possible expansions of an elementary trapping set $\mathcal{S}$ to a larger elementary trapping set $\mathcal{S}'$. (Unsatisfied check nodes of $G({\mathcal{S}'})$ are not shown.)}
\label{fig3}
\end{figure}

\begin{lemma}
Suppose that
$\ A=\{a_1,...,a_i,a_{i+1},...,a_I\}$ is the sorted
set of sizes of the elementary trapping sets in $\mathcal{T}$ in increasing order. Let $\mathcal{S}$ be an elementary trapping set of size $a_i$ in $\mathcal{T}$. If $a_{i+1} > a_i+2$,
then for each elementary trapping set $\mathcal{S}'\in \mathcal{T}_{\mathcal{S}}^{a_{i+1}}$ (if any),
the set $\mathcal{S}'\backslash \mathcal{S}$ is connected to zero satisfied
check nodes of $\mathcal{S}$ (i.e., nodes in $\Gamma_{\mathrm{e}}(\mathcal{S})$) and to only one or two unsatisfied
check nodes of $\mathcal{S}$ (i.e., nodes in $\Gamma_{\mathrm{o}}(\mathcal{S})$). If the set $\mathcal{S}'\backslash \mathcal{S}$ is connected to
two nodes in $\Gamma_{\mathrm{o}}(\mathcal{S})$, then there is no cycle in $G({\mathcal{S}'\backslash \mathcal{S}})$.
If the set $\mathcal{S}'\backslash \mathcal{S}$ is connected to only one node in $\Gamma_{\mathrm{o}}(\mathcal{S})$, then
there is exactly one cycle in $G({\mathcal{S}'\backslash \mathcal{S}})$.
\label{lem3}
\end{lemma}

\begin{proof}
If any variable node in $\mathcal{S}'\backslash \mathcal{S}$ is connected to $\Gamma_{\mathrm{e}}(\mathcal{S})$, then $G({\mathcal{S}')}$ contains
satisfied check nodes of degree $4$ or higher, or unsatisfied check nodes of degree greater than $1$. Both
are in contradiction with $\mathcal{S}'$ being an elementary trapping set.

Since $\mathcal{S}'$ is an elementary trapping set, there cannot be more than one connection between
each node of $\Gamma_{\mathrm{o}}(\mathcal{S})$ and the nodes in $\mathcal{S}'\backslash \mathcal{S}$.
To see this, consider a node $v \in \mathcal{S}'\backslash \mathcal{S}$ which is connected to $\Gamma_{\mathrm{o}}(\mathcal{S})$. Node $v$ can have only
one connection to $\Gamma_{\mathrm{o}}(\mathcal{S})$ because otherwise, all the other nodes in $\mathcal{S}'\backslash \mathcal{S}$ can be
removed and we will end up with an elementary trapping set of size $a_i+1$ in $\mathcal{T}$, which is in contradiction
with the assumption of the lemma.

Now suppose that there are at least 3 connections
between the variable nodes in $\mathcal{S}'\backslash \mathcal{S}$ and $\Gamma_{\mathrm{o}}(\mathcal{S})$. Based on the discussion in the previous paragraph, this means that there are at least $3$ variable nodes in $\mathcal{S}'\backslash \mathcal{S}$ each with a single connection to a different check node in $\Gamma_{\mathrm{o}}(\mathcal{S})$. If $G({\mathcal{S}'\backslash \mathcal{S}})$ is not connected,
then one can remove one of its components and obtain an elementary trapping set of
size smaller than $a_{i+1}$, which results in a contradiction. If $G({\mathcal{S}'\backslash \mathcal{S}})$ is connected, then one can
find the shortest paths in $G({\mathcal{S}'\backslash \mathcal{S}})$ between every two variable nodes of $\mathcal{S}'\backslash \mathcal{S}$
that are connected to $\Gamma_{\mathrm{o}}(\mathcal{S})$, and among them select the one with the least number of nodes. By keeping the nodes on
the selected path and removing all the other nodes in $\mathcal{S}'\backslash \mathcal{S}$, one can then
obtain an elementary trapping set of size smaller than $a_{i+1}$, which is again a contradiction.
We therefore conclude that the number of connections between the variable nodes
in $\mathcal{S}'\backslash \mathcal{S}$ and $\Gamma_{\mathrm{o}}(\mathcal{S})$ must be strictly less than 3.

For the case that $\mathcal{S}'\backslash \mathcal{S}$ is connected to exactly two nodes in $\Gamma_{\mathrm{o}}(\mathcal{S})$,
there must be two different variable nodes $v$ and $v'$ of $\mathcal{S}'\backslash \mathcal{S}$ corresponding to those connections. Also,
there must be no cycles in $G({\mathcal{S}'\backslash \mathcal{S}})$. Otherwise, one can remove all the variable nodes on the cycle except those on the shortest path between $v$ and $v'$, and obtain an elementary trapping set larger than $a_i$ but smaller than $a_{i+1}$.
This contradicts the lemma's assumption. Fig.~\ref{fig3}$(b)$ is an example of the case where $\mathcal{S}'\backslash \mathcal{S}$ is connected to exactly two nodes in $\Gamma_{\mathrm{o}}(\mathcal{S})$.

The proof for the case with one connection is similar and omitted. Fig.~\ref{fig3}$(c)$ is an example of this case, where the expansion of set $\mathcal{S} $ is through a lollipop walk. In both Figs.~\ref{fig3}$(b)$ and $(c)$, the dashed line indicates that more variable and check nodes can be part of the chain.
\end{proof}

\begin{lemma}
Suppose that $\ A=\{a_1,...,a_i,a_{i+1},...,a_I\}$ is the sorted
set of sizes of the elementary trapping sets in $\mathcal{T}$ in increasing order and that $a_{i+1}= a_i+2$. Let $\mathcal{S} $ be an elementary trapping set of size $a_i$ in $\mathcal{T}$.
If the girth of the graph is larger than 4, then for each elementary trapping
set $\mathcal{S}'\in \mathcal{T}_{\mathcal{S}}^{a_{i+1}}$ (if any), the only possible configuration for $G({\mathcal{S}'})$ is that
of Fig.~\ref{fig3}$(b)$, described in Lemma~\ref{lem3}, with only 2 variable nodes in $\mathcal{S}'\backslash \mathcal{S}$. 
If the girth
is 4, then the only possible configurations are those in Figs.~\ref{fig3}$(b)$ (with only 2 variable nodes in $\mathcal{S}'\backslash \mathcal{S}$) \ref{fig3}$(d)$ and \ref{fig3}$(e)$. 
\label{lem4}
\end{lemma}

\begin{proof}
The proof is similar to that of Lemma~\ref{lem3} and is omitted.
\end{proof}

\subsection{Proposed Algorithm}
The basic idea behind the proposed algorithm is to construct larger elementary trapping sets
by expanding smaller ones. More precisely, given an elementary trapping set $\mathcal{S} $ of size $a_i$ at the input,
the algorithm finds all the elementary trapping sets $\mathcal{S}'$ containing $\mathcal{S} $, with the property
that their size $a_{i+1}$ is the smallest size greater than $a_i$. The algorithm then continues
by using the sets found in the current step as the inputs to the next step and finds the next
set of larger elementary trapping sets. Each step of the algorithm is performed
by using Lemmas~\ref{lem2} - \ref{lem4}. The pseudo-code for one step of the proposed algorithm
is given in Algorithm 1.

\hspace{-.25in}
\linethickness{0.275mm}
\line(1,0){250} \\
\textbf{Algorithm 1}: Expansion of input elementary trapping sets to larger ones of size up to $k$ with the number of unsatisfied check nodes up to $T$ in $G=(L\cup R\,,E)\,$.\\
($\mathcal{L}_{in}$ and $\mathcal{L}_{out}$ are the lists of input and output trapping sets, respectively.)\\
\linethickness{0.125mm}
\line(1,0){250}
\begin{algorithmic}[1]
\STATE  \textbf{Inputs:} $G$ and $\mathcal{L}_{in}$.
\STATE  \textbf{Initialization:} $\mathcal{L}_{out}$ $\leftarrow \emptyset$.

\REPEAT
\STATE Select an element of $\mathcal{L}_{in}$ and denote it as  $t_j$.
\STATE Construct a new graph $G'$ by removing all the nodes in $\Gamma_{\mathrm{e}}(t_j)$ and their neighbors from $G$.
\STATE $i_{max} \leftarrow$ $(k-|t_j|)$ and $\mathcal{G}\leftarrow \emptyset$.
\label{step1}
\FOR {each node $c$ in $\Gamma_{\mathrm{o}}(t_j)$}
\STATE Examine the neighborhood of $c$ in $G'$ one layer at a time and to the maximum of $i_{max}$ layers in search for paths with $i\le i_{max}$ variable nodes between $c$ and the other nodes of $\Gamma_{\mathrm{o}}(t_j)$, and lollipop walks
with $i\le i_{max}$ variable nodes starting from $c$.
\STATE Denote $\mathcal{G}_c$ as the set of all such paths/lollipop walks of shortest length $i$ (if any).
\IF {$i<i_{max}$}
\STATE $i_{max} \leftarrow$ $i$.
\STATE $\mathcal{G}\leftarrow$ $\mathcal{G}_c$.
\ELSE
\STATE $\mathcal{G}\leftarrow \mathcal{G} \cup \mathcal{G}_c$.
\ENDIF
\ENDFOR
\FOR {each element $\mathcal{S} $ in $\mathcal{G}$}
\STATE $t'\leftarrow t_j \cup \mathcal{S}$.
\IF {($t' \notin \mathcal{L}_{out}$) and ($|\Gamma_{\mathrm{o}}(t')|\le T$) }
\label{step2}
\STATE $\mathcal{L}_{out}\leftarrow \mathcal{L}_{out} \cup \{t'\}$.
\ENDIF
\ENDFOR
\UNTIL { all the elements of $\mathcal{L}_{in}$ are selected.}
\STATE  \textbf{Output:} $\mathcal{L}_{out}$.
\end{algorithmic}
\linethickness{0.275mm}
\line(1,0){250}

\begin{remark}
Note that in Line 5 of Algorithm~1, all the satisfied check nodes in $G({t_j})$,
i.e., the set $\Gamma_{\mathrm{e}}(t_j)$, and their neighboring variable nodes are removed from the graph.
This is because, based on Lemmas~\ref{lem2}~-~\ref{lem4}, such nodes cannot be part of the
expansion of an elementary trapping set.
\end{remark}
\begin{remark}
In Line 19 of the algorithm,
the threshold value $T$ on the number of unsatisfied check nodes is needed to
keep the complexity of the overall search algorithm, which involves multiple applications of Algorithm 1, low. A proper choice of $T$ has negligible effect on the ability of the algorithm to find the larger trapping sets $(a,b)$ with small values of $b$. This is explained in the following example.

\begin{example}
\begin{figure}[!h]
\centering
\includegraphics[width=2in]{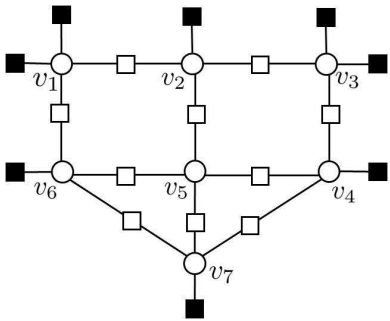}
\caption{An example of a $(7,\,8)$ trapping set (satisfied and unsatisfied check nodes are shown by empty and full squares, respectively)}
\label{fig5}
\end{figure}
Consider the $(7,\,8)$ trapping set $\mathcal{S} $, shown in Fig.~\ref{fig5}. This set belongs to $\mathcal{T}$ and contains $4$ trapping sets of size $6$, all also in $\mathcal{T}$.
These four trapping sets can be each obtained
by removing one of the nodes $v_1$, $v_3$, $v_5$ and $v_7$. As a result, we have $(6,\,8)$, $(6,\,8)$, $(6,\,12)$ and $(6,\,10)$
trapping sets, respectively. Among these trapping sets, the $(6,\,8)$ ones have a smaller number of unsatisfied check nodes.
Starting from each of these two trapping sets, Algorithm $1$ finds $\mathcal{S} $.
Hence, ignoring the $(6,\,12)$ (or even $(6,\,12)$ and $(6,\,10)$) trapping set(s) does not impair the ability of the algorithm to find $\mathcal{S} $.
\end{example}

\end{remark}
\begin{remark}
Based on Lemmas~\ref{lem2}~--~\ref{lem4}, it can be proved that starting
from an $(a,b)$ elementary trapping set $\mathcal{S} $, Algorithm~1 will find all
the $(a',b')$ elementary trapping sets of the smallest size $a'$ larger than $a$
that contain $\mathcal{S} $ (this requires the removal of the condition $|\Gamma_{\mathrm{o}}(t')|\le T$ in Line 19).
Note that this does not imply that by the recursive
application of Algorithm~1 one can obtain all the elementary trapping sets
containing $\mathcal{S} $. The following example demonstrates this.

\begin{example}
Consider the $(6,6)$ elementary trapping
set ${\cal S}'=\{v_1,\,v_2,\,v_3,\,v_4,\,v_5,\,v_6\}$ in Fig.~\ref{fig4_1}. Assume that Algorithm 1
starts from the elementary trapping set ${\cal S} =\{v_1,\,v_2,\,v_6\}$. Using this input, the output of the algorithm
is $\{v_1,\,v_2,\,v_6,\,v_5\}$. By subsequent applications of the algorithm, the next outputs are $\{v_1,\,v_2,\,v_6,\,v_5,\,v_7\}$
and $\{v_1,\,v_2,\,v_6,\,v_5,\,v_7,\,v_3,\,v_4\}$, respectively. This means that
the algorithm does not find the trapping set ${\cal S}'$, although ${\cal S}'$ contains ${\cal S}$. (It is however easy to see that if the algorithm starts from the set $\{v_2,\,v_3,\,v_4,\,v_5\}$, it will find ${\cal S}'$.)
\begin{figure}[!h]
\centering
\includegraphics[width=2in]{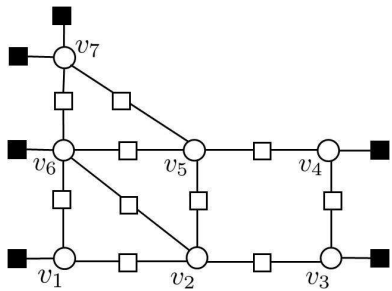}
\caption{An example explaining that the algorithm cannot find all the elementary trapping sets containing a specific elementary trapping set }
\label{fig4_1}
\end{figure}
\label{EXample3}
\end{example}
In fact, the sufficient condition for the
algorithm to find a trapping set $\mathcal{S}'$ of size $a_j$, starting with one of its subsets $\mathcal{S} $ of size $a_i<a_j$, is that $\mathcal{S}'$ has at least
one subset in $\mathcal{T}_{\mathcal{S}}^a$ for all $a\in A, \,\,a_i<a<a_j$, where $A$ is defined in Lemma~\ref{lem3}.
\label{remark3}
\end{remark}

The following example shows that despite the limitation explained in Remark~\ref{remark3} and Example~\ref{EXample3}, for many cases, the proposed algorithm in fact finds (in a guaranteed fashion) \emph{all} the trapping sets $(a,b)$ with $a$ and $b$ up to certain values.

\begin{example}
For a left-regular graph with left degree 4 and girth larger than 4, initiating the proposed algorithm with the set of cycles of length $g$ and $g+2$, one can guarantee to find \emph{all} the elementary trapping sets of size less than $9$ with less than $5$ unsatisfied check nodes. This can be seen by the inspection of all the possible structures for such trapping sets and verifying that for each structure, the removal of only one variable node will result in another trapping set in $\mathcal{T}$. Subsequent removals of such nodes from an $(a,b)$ elementary trapping set with $a < 9$ and $b < 5$ in ${\cal T}$, thus, results in a sequence of embedded elementary trapping sets in ${\cal T}$, each with size only one less than that of its parent. The sequence will always end with a cycle of length $g$ or $g+2$. This implies that all such trapping sets satisfy the sufficient condition, mentioned earlier, for being found by the algorithm starting from a cycle of length $g$ or $g+2$.
\label{example4}
\end{example}

Similar results to those of Example~\ref{example4} can be found for other left-regular graphs. For irregular graphs however, it is very difficult to provide such guarantees. This is due to the fact that the number of possible structures for a trapping set of a given size could be very large in this case.

\begin{remark}
For irregular LDPC codes, the variable nodes with large degrees cannot be part of small trapping sets. This is formulated in the following lemma.
\begin{lemma}
In a graph $G$ with girth $g>4$, if an $(a,b)$ trapping set $\mathcal{S} $ contains a variable node $v$ of degree $d(v)>b$, then $a\ge d(v)+1-b$.
\label{lem99}
\end{lemma}
\begin{proof}
The proof is provided in Appendix A.
\end{proof}
Based on Lemma \ref{lem99}, for example, for an irregular code with girth larger than $4$, a variable node of degree 15 can not participate in an $(a,\, b)$ trapping set with $a < 13$ and $b < 4$. Such results can be used to simplify the algorithm by removing the large degree variable nodes and their edges from the graph.
\end{remark}
\begin{remark}
It is easy to see that for the left-regular graphs with left degree 3 or 4, all the trapping sets found by Algorithm 1 are ZP trapping sets. For the left-regular graphs with left degree 3, the obtained trapping sets are also absorbing sets.
\end{remark}
\begin{remark}
Our simulations for many practical LDPC codes show that in almost all the cases, $a_{i+1}\le a_i+3$.
\end{remark}

In the following, we discuss the selection of the initial set of elementary trapping sets.

\subsection{Initial Set of Trapping Sets}

One of the graphical objects that plays an important role in the structure of trapping sets
is a cycle. Tian {\em et al.}~\cite{Tian2004} showed that every stopping set includes the variable nodes
of at least one cycle. Related to this, the induced graph of the support of a pseudo-codeword always contains at least one cycle \cite{KV-03}.
In~\cite{Cole2008},~\cite{Xiao2009},~\cite{Zhang2009}, it was shown
that an overwhelming majority of dominant trapping (absorbing) sets are combinations of short cycles.
Short cycles are also easy to enumerate~\cite{Xiao2009}.
We thus use short cycles as the initial inputs to the proposed algorithm.
The following lemma provides more justifications for this choice.

\begin{lemma}
\begin{itemize}
\item []
\item [i)] In a left-regular graph $G$ with left degree $d_l\ge2$, if the induced subgraph $G(\mathcal{S})$ of an $(a,\,b)$ trapping set $\mathcal{S} $ does not contain any cycle, then $b\ge a(d_l-2)+2$. The inequality is satisfied with equality for elementary trapping sets.
\label{lem10}
\item [ii)] The variable nodes in any shortest cycle (of length $g$) of a Tanner graph
form an elementary trapping set.
\label{lem5}
\item [iii)] Let $\mathcal{T}$ be the set of all trapping sets $\mathcal{S} $ of a graph $G$, whose induced
subgraph $G(\mathcal{S})$ is connected and for which every node $v \in \mathcal{S}$ is connected
to at least two nodes in $\Gamma_{\mathrm{e}}(\mathcal{S})$. Then for every $\mathcal{S} \in \mathcal{T}$,
its induced subgraph $G(\mathcal{S})$ contains at least one cycle.
\label{lem6}
\item [iv)] Suppose that $\mathcal{S} \subset L$ is an absorbing set of a left-regular Tanner graph $G=(L\cup R,E)$
with left node degrees at least 2. Then $G(\mathcal{S})$ contains at least one cycle.
\label{lem7}

\item [v)] Suppose that $\mathcal{S} \subset L$ is a ZP trapping set of a Tanner graph $G=(L\cup R,E)$
with node degrees at least 3. Then $G(\mathcal{S})$ contains at least one cycle.
\label{lem8}
\end{itemize}
\end{lemma}
\begin{proof}
The proof of Part (i) is provided in Appendix A. The proofs for Parts (ii) and (iii) are simple and thus omitted. Parts (iv) and (v) follow from Lemma~\ref{lem1} and Part (iii).
\end{proof}

It can be shown that Part (i) of Lemma~\ref{lem8} can be generalized to the case where variable node degree distribution is irregular. In this case, the result is modified as $b\ge a(\bar{d_{\mathcal{S}}}-2)+2$, where $\bar{d_{\mathcal{S}}}$ is the average degree of variable nodes in $\mathcal{S} $.  The following example, based on Part (i) of Lemma~\ref{lem8}, demonstrates that cycle-free $(a,b)$ trapping sets have relatively large values of $b$.

\begin{example}
For a left-regular graph $G$ with left degree 4, any cycle-free $(a,\,b)$ trapping set satisfies $b\ge 2(a+1)$. Such large values of $b$ for a given $a$ would imply that the $(a,\,b)$ trapping set is not dominant.
\end{example}

%
%
%

Our simulation results indicate that for denser graphs, the set of short cycles of
length $g$, or $g$ and $g+2$, where $g$ is the girth, is sufficient to find
almost all the small (with, say, $a \leq 10$) dominant trapping sets. In this case, adding short cycles
of larger lengths to the input set has negligible effect on the performance of the algorithm, while increasing its complexity. For example, we examined a number of randomly constructed codes with rates larger than 0.4. The codes had block length 1000 and left-regular Tanner graphs with left degree 5 and girth 6. In all cases, the trapping sets obtained by Algorithm 1 using cycles of length 6 and 8 as input were identical to those obtained by using cycles of length 6, 8 and 10 as the input set.

For sparser graphs, however, one may need to use short cycles of larger lengths (e.g., $g$, $g+2$, and $g+4$) as
the initial set.
\subsection{Complexity of the Algorithm}
The complexity of the algorithm is highly dependent on the short cycle distribution of the graph,
which itself is mostly a function of the degree distribution of the graph (code) \cite{Karimi2010}. As a result, in general, the complexity
increases much faster with the increase in the average variable and check node degrees of the graph than it does with increasing the block length.
To have a more detailed analysis of the complexity of Algorithm~1, we note that the total complexity can be divided into two parts: $a)$ Finding the initial input set and $b)$ Expanding the input set to larger trapping sets.

Regarding the complexity of Part $(a)$, assuming that an exhaustive brute force search is used to find cycles of length $k$, say for $g \leq k \leq g+4$, the complexity is  $O(nd_v^{k/2}d_c^{k/2})$,  for a $(d_v,d_c)$ regular graph with $n$ variable nodes. This is obtained by considering all the possible paths of length $k$ starting from all the $n$ variable nodes in the graph. The memory required for the storage of all the $k$-cycles is of order $O(kN_k)$, where $N_k$ is the number of $k$-cycles in the graph.
To the best of our knowledge, there is no theoretical result on how $N_k$ scales with $n$ or the degree distribution of the Tanner graph. Empirical results of \cite{Karimi2010} however suggest that $N_k$ is mainly a function of the degree distribution and is rather independent of $n$.

Regarding the complexity of expanding the input trapping sets to larger ones, consider the expansion of an $(a,b)$ trapping set ${\cal S}$ of a $(d_v,d_c)$ graph. Depending on the size $a' > a$ of the smallest trapping set(s) ${\cal S}'$ that contain ${\cal S}$, the complexity and memory requirements for finding and storing the sets ${\cal S'}$ would differ. For $a' = a+1, a+2$ and $a +3$, the complexity is $O(b d_c), O(b d_v d_c^2)$ and $O(b d_v^2 d_c^3)$, respectively. The memory requirement for these cases are respectively $O(a b d_c), O(a b d_v d_c^2)$ and $O(a b d_v^2 d_c^3)$. To see this, for example, consider the case where $a'=a+1$. To find ${\cal S}'$, one needs to check at most $b (d_c - 1)$ variable nodes as possible candidates, which corresponds to $O(b d_c)$ complexity. The memory required to store all possible trapping sets of size $a+1$ obtained through such a search is thus upper bounded by $(a+1) b d_c$, which is of order $O(a b d_c)$.

Based on the above discussions, assuming that the initial set is limited to cycles of length up to $g+4$, and that we only consider trapping sets ${\cal S}'$ of size up to $a' = a +3$ in the expansion process, the complexity of Algorithm~1 will be $O(d_v^2 d_c^3 (T \sum_{i=g/2}^{g/2+2} N_{2i} + n d_v^{g/2} d_c^{g/2-1}))$ and the memory requirement will be $O(T d_v^2 d_c^3 \sum_{i=g/2}^{g/2+2} 2i N_{2i})$, where $T$ is the maximum number of unsatisfied check nodes in Algorithm~1. It is however important to note that the actual complexity and memory requirements are much less than what these complexity bounds may suggest. In particular, our simulation results show that codes with block lengths up to about 10,000 with a wide variety of
degree distributions can be managed using the proposed algorithm on a regular desktop computer.

\subsection{Expansion of Non-Elementary Trapping Sets}

According to the general definition of trapping sets, any arbitrary set of variable nodes can be considered as a trapping set. Hence, to expand a connected trapping set $\mathcal{S} $ of size $a$, one just needs to select a variable node from the neighboring variable nodes, and add it to $\mathcal{S} $ to obtain a new trapping set $\mathcal{S}'$ with size $a'=a+1$. This method of expansion leads to an exponentially growing search space. Even by limiting the search space to the trapping sets in $\mathcal{T}$, i.e., connected trapping sets for which every variable node is connected to at least two satisfied check nodes, there are still too many configurations for $\mathcal{S}'$, especially when $a'\gg a$. For practical LDPC codes with $g>4$, however, considering a nested sequence of trapping sets, the size of the next larger trapping set $a'$ is almost always less than $a+3$.

The search for non-elementary trapping sets of size $a'\le a+3$ in a graph with girth $g>4$, can be performed similar to what was described for the elementary trapping sets with a number of small differences. For non-elementary trapping sets, since there is no limitation on the degrees of the check nodes in $G(\mathcal{S})$, only the variable nodes of $\mathcal{S} $ and their edges are removed from the graph. Then the shortest paths between different check nodes of $G(\mathcal{S})$ or the shortest lollipop walks starting from different check nodes of $G(\mathcal{S})$ are found. However, it should be mentioned that not all such structures will necessarily satisfy the condition that each variable node is connected to at least two satisfied check nodes. After finding a candidate trapping set, one should thus check for this condition. In summary, to find the non-elementary trapping sets of size $a'\le a+3$, the only modifications needed to be applied to Algorithm 1 are the followings:
\begin{algorithmic}
\STATE 5: Construct a new graph $G'$ by removing all the nodes of $\mathcal{S} $ from $G$.
\STATE 7: \textbf{for}  each node $c$ in $\Gamma(t_j)$ \textbf{do}\\
\STATE 8: Examine the neighborhood of $c$ in $G'$ one layer at a time and to the maximum of $i_{max}$ layers in search for paths with $i\le i_{max}$ variable nodes between $c$ and the other nodes of $\Gamma(t_j)$, and lollipop walks
with $i\le i_{max}$ variable nodes starting from $c$.
\STATE 19: \textbf{if} {($t' \in \mathcal{T}$) and ($t' \notin \mathcal{L}_{out}$) and ($|\Gamma_{\mathrm{o}}(t')|\le T$) }
\end{algorithmic}


\section{irregular LDPC codes}
For the irregular LDPC codes which do not have variable nodes of degree 2, Algorithm 1 without any modification can be used to find the dominant trapping sets. As mentioned in Remark 4 of Section III.C, based on the desired sizes of trapping sets, one may also remove the high-degree variable nodes and their edges from the graph to simplify the algorithm. In the case that the code has variable nodes of degree 2, some modifications are needed for the initial input set of the algorithm. In this section, we study the effect of degree-2 variable nodes on the structure of trapping sets in irregular LDPC codes, and present simple steps to find the corresponding trapping sets.\footnote{In case that the graph contains degree-1 variable nodes as well, a similar approach to the one described in Section IV.B (for finding dominant trapping sets which include degree-2 variable nodes) can be used to find the dominant trapping sets containing degree-1 variable nodes.}

\subsection{On the Degree-2 Variable Nodes}
It is known that degree-2 variable nodes play an important role in the performance of irregular LDPC codes. On one hand, to have codes with asymptotic performance close to the capacity, the proportion of degree-2 variable nodes should be as large as possible. This is usually a considerable fraction of the total variable nodes of the code. On the other hand, having a large proportion of degree-2 variable nodes results in a  small minimum distance 
and  a high error floor \cite{Tillich2006}. Cycles containing only degree-2 variable nodes are codewords. Hence, to have a large minimum distance, it is desirable to avoid such cycles, especially the shorter ones. To avoid all cycles of any length containing only degree-2 variable nodes, the number of these nodes $n_{v_2}$ must be strictly less than the number of check nodes $m$ (i.e., $n_{v_2}< m$). Based on this fact, a class of irregular LDPC codes with $n_{v_2}=m-1$, called \emph{extended irregular repeat accumulate} (eIRA) codes was proposed in \cite{Yang2004}. It was shown in \cite{Yang2004} that these codes exhibit relatively better error floor performance compared to the codes constructed by the optimized degree distributions without applying this restriction on $n_{v_2}$. Related to this, it was proved in \cite{Otmani2007} that for the case where $n_{v_2}>m$, the minimum distance grows only logarithmically with the code length. For the special case where $n_{v_2}=m$ and all the degree-2 variable nodes are part of a single cycle, the minimum distance is a sub-linear power function of the block length \cite{Tillich2006}. In the following, we study the effect of having a large fraction of degree-2 variable nodes on the structure of trapping sets in irregular LDPC codes.

\begin{example}
  For all the degree distributions optimized  for rate-1/2 LDPC codes on the binary-input AWGN (BIAWGN) channel  \cite{Richardson2001}, $43\%$ to $55\%$ of variable nodes are of degree 2. This implies that, on average,  every check node in the corresponding codes is connected to about $2$ variable nodes of degree 2.
\end{example}
The average number of degree-2 variable nodes connected to each check node becomes even larger for the optimized codes of higher rate. This is explained in the next example.
\begin{example}
   For the optimized degree distribution of rate $8/9$ over the BIAWGN channel with the maximum variable node degree 10  \cite{Richardson2001}, $31\%$ of variable nodes are of degree 2. This implies that, on average, every check node in a Tanner graph with this degree distribution is connected to about $6$ variable nodes of degree 2.
\end{example}

Consequently, it is very likely to see chains of degree-2 variable nodes, referred to as \emph{2-chains}, in the Tanner graph of LDPC codes with optimized degree distributions. The length of a 2-chain is defined as the number of the edges in the subgraph induced by the degree-2 variable nodes of the chain. That is, the length of a 2-chain containing $k$ variable nodes of degree 2 is $2k$. A 2-chain of length $2k$ is a $(k,\,2)$ trapping set (with the exception of the case where the chain is closed and forms a cycle; in that case, we refer to the 2-chain as a \emph{2-cycle}. A 2-cycle of length $2k$, is a $(k,\,0)$ trapping set). Having only 2 unsatisfied check nodes, 2-chains of length $2k$ are among the most dominant trapping sets of size $k$. Fig.~\ref{degree2}$(a)$ shows a 2-chain of length 10 (a $(5,\,2)$ trapping set). Note that this trapping set also contains  two $ (4,\,2)$, three $ (3,\,2)$ and four $(2,2)$ trapping sets as its subsets. It is worth noting that although for the cases where $n_{v_2}=m-1$ and  $n_{v_2}=m$, the graph may have no or only one 2-cycle, it can have many 2-chains of different lengths. For example, it is easy to see that for the case where $m=n_{v_2}$ and all the degree-2 variable nodes are contained in a single cycle, there are $m$ 2-chains of length $2k$, $1\le k \le m-1$.

\begin{figure}[!t]
\centering
\includegraphics[width=1.6in]{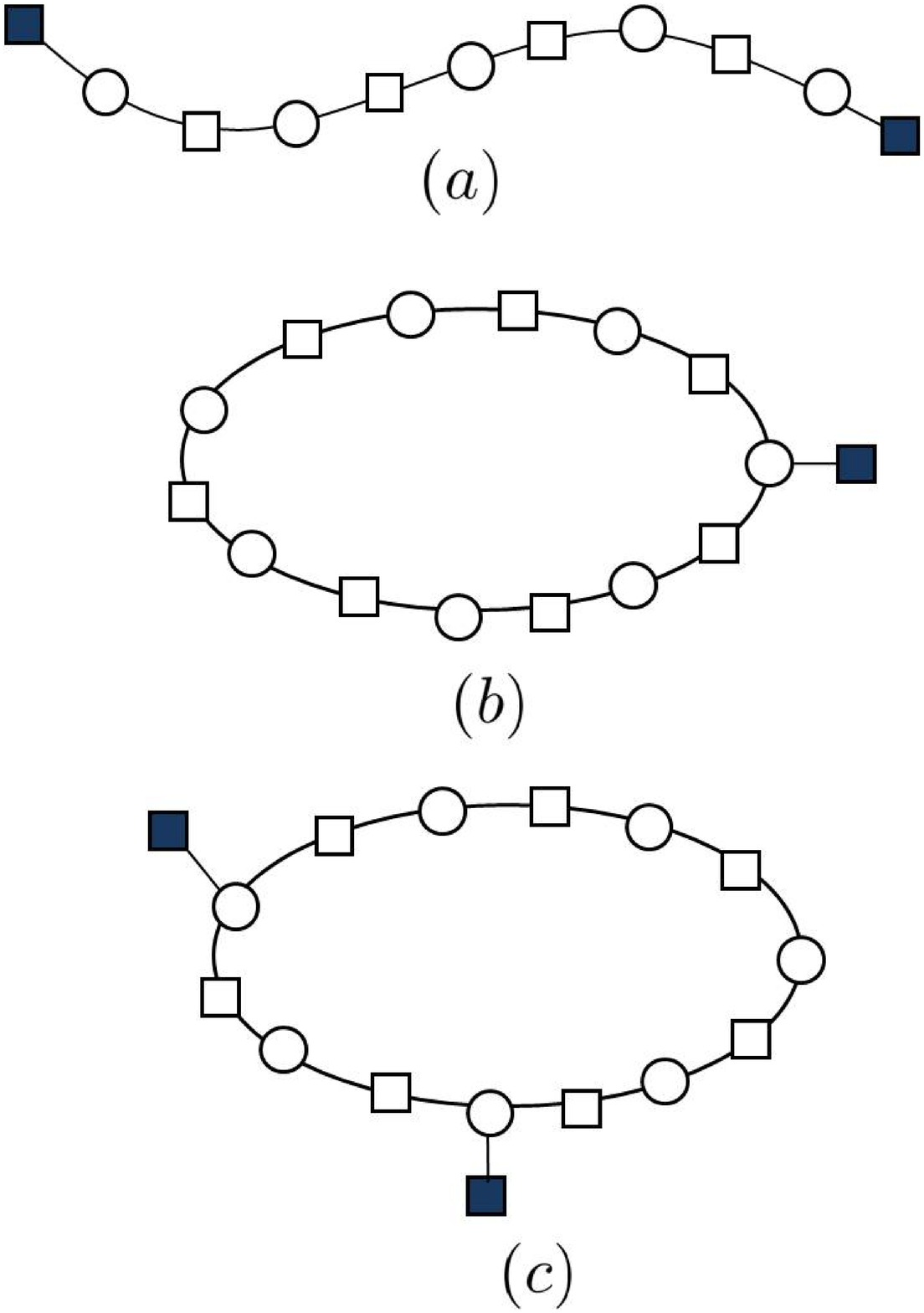}
\caption{Typical trapping sets constructed mostly by the degree-2 variable nodes.}
\label{degree2}
\end{figure}

Another aspect of having  2-chains in the Tanner graph of irregular LDPC codes is that they might participate in short cycles with other variable nodes of higher degrees. These cycles have low approximate cycle extrinsic message degree (ACE) (ACE is defined as $\sum_i d_i-2 $, where the summation is taken over all the variable nodes of the cycle, and $d_i $ is the degree of the $i$$^{th}$ variable node in the cycle \cite{Tian2004}). It has been shown that cycles with low ACE deteriorate the error rate performance, and that avoiding them in the construction of irregular LDPC codes generally improves the error rate \cite{Xiao2004}, \cite{Vukobratovic2008}.
\begin{example}
Consider the case where $m=n_{v_2}$ and all the degree-2 variable nodes are contained in a single cycle. In this case, there exist two 2-chains between any two check nodes of the graph. This implies that every variable node of degree $d_v>2$ along with the 2-chains connecting its check nodes form several trapping sets with at most $d_v-2$ unsatisfied check nodes.
\end{example}
\begin{example}
 Fig.~\ref{degree2}$(b)$ shows a $(7,\,1)$ trapping set composed of one variable node of degree 3 and a chain of six variable nodes of degree 2.
 \end{example}
 \begin{example}
The $(12,\,1)$ trapping sets of the $(1944,\,972)$ LDPC code adopted in the IEEE 802.11 standard \cite{IEEE802.11n} are single cycles of length $24$, each consisting of a 2-chain of length 22 and one degree-3 variable node.
 \end{example}
Even in the cases where $n_{v_2}<m$ (but not much smaller), it is likely to see cycles mostly constructed by 2-chains.
\begin{example}
Fig.~\ref{degree2}$(c)$ shows a $(7,\,2)$ trapping set composed of two variable nodes of degree 3 and five variable nodes of degree 2.
\end{example}
  Due to the important role that 2-chains (and 2-cycles) play in the formation of dominant trapping sets, we study the necessary condition to avoid these structures in the following theorem.

\begin{theorem}
Let $m$ be the number of check nodes and $n_{v_2}$ be the number of degree-2 variable nodes in the graph $G$ corresponding to an irregular code $\mathcal{C}$. If $G$ has no 2-chains of length $2k$ or larger, for $k \geq 2$ (and no 2-cycles of length less than or equal to $2k$) then
\[m\ge n_{v_2}(1+\frac{1}{\sum_{i=0}^{k-2}{(d_{c,max}-1)^{\lfloor\frac{i+1}{2} \rfloor} } })\,,\]
where $d_{c,max}$ is the maximum check node degree in $G$.
\label{nv2_m}
\end{theorem}
\begin{proof}
Let $G_{v_2}$ denote the induced subgraph of degree-2 variable nodes of the graph $G$. This subgraph contains no cycle. Otherwise, the length of such a cycle would be at least $2k+2$, which would imply the existence of a 2-chain of length $2k$ in $G_{v_2}$, and thus in $G$. This contradicts the assumption of the theorem. The subgraph $G_{v_2}$ is thus composed of some tree-like components. For each component, the number of check nodes is always larger than the number of variable nodes by one. Therefore the total number of check nodes of the graph is more than the number of degree-2 variable nodes by at least the number of disjoint components in  $G_{v_2}$ (some check nodes of $G$ may not appear in  $G_{v_2}$). To avoid 2-chains of length $2k$ or larger, the maximum number of variable nodes in each component is $\sum_{i=0}^{k-2}{(d_{c,max}-1)^{\lfloor\frac{i+1}{2} \rfloor}} $(Appendix A, Lemma \ref{2chain_size}). The minimum number of components in $G_{v_2}$ is thus $\lceil n_{v_2}/\sum_{i=0}^{k-2}{(d_{c,max}-1)^{\lfloor\frac{i+1}{2} \rfloor}} \rceil$ .
\end{proof}
 \vspace{.2cm}

 Theorem~\ref{nv2_m} can be used to determine the maximum number of degree-2 variable nodes in an irregular graph to avoid 2-chains (and 2-cycles) of a specific length.
 \begin{example}
   For an irregular code with 1000 check nodes of degree $d_c=6$, to avoid $(4,\,2)$ trapping sets corresponding to 2-chains of length 8, the number of variable nodes of degree 2 must be  at most 910.
 \label{ex1000}
 \end{example}
 Theorem~\ref{nv2_m} can be also used to obtain some information about the existing trapping sets in a code.
 \begin{example}
 For the same scenario as that of Example~\ref{ex1000} (i.e., $m=1000,\,\, d_c=6$), the eIRA construction \cite{Yang2004} results in $n_{v_2}=m-1=999$. For these parameters, the smallest value of $k$ which satisfies the inequality of Theorem 1 is $k=9$. This implies that the eIRA code will have 2-chains of length 16 and smaller, corresponding to $(k,\,2)$ trapping sets for all values of $k<9$.
 \end{example}
\subsection{Finding Trapping Sets of Irregular LDPC Codes}
\label{abf}
In this section, we present a simple process to find the dominant trapping sets involving degree-2 variable nodes. The process can be used in combination with Algorithm 1 to find the dominant trapping sets of irregular graphs containing degree-2 variable nodes. It is important to note that according to the definition of absorbing sets, any variable node of degree 2 in these sets is connected to 2 satisfied check nodes. Also, for the trapping sets found by Algorithm 1, each variable node is connected to at least 2 satisfied check nodes. Therefore, 2-chains and other trapping sets containing variable node(s) of degree 2 with one satisfied check node are neither absorbing sets nor found by Algorithm 1. In fact, it appears that being connected to 2 satisfied check nodes is too strong of a condition for a variable node of degree 2 to be part of a dominant trapping set. For this reason, we consider also trapping sets whose variable nodes of degree 2 are connected to only one satisfied check node. To obtain such trapping sets using the expansion of smaller trapping sets, we consider an $(a-1,\,b)$ trapping set $\mathcal{S} $ which is expanded to a trapping set $\mathcal{S}'$ by the connection of a variable node $v$ of degree 2 to an unsatisfied check node of $\mathcal{S} $. Three cases are possible:
\begin{itemize}
\item [$a)$]$v$ is not connected to any other check node of $\Gamma(\mathcal{S})$. In this case, $\mathcal{S}'=\mathcal{S}\cup\{v\}$ is an $(a,b)$ trapping set. If $\mathcal{S} $ is elementary, so is $\mathcal{S}'$.
\item [$b)$]$v$ is also connected to a satisfied check node of $\mathcal{S} $. In this case, $\mathcal{S}'=\mathcal{S}\cup\{v\}$ is an $(a,b)$ non-elementary trapping set.
\item [$c)$]$v$ is also connected to another unsatisfied check node of $\mathcal{S} $. In this case, $\mathcal{S}'=\mathcal{S}\cup\{v\}$ is an $(a,b-2)$ trapping set. If $\mathcal{S} $ is elementary (or is in the set $ \mathcal{T}$), so is $\mathcal{S}'$.
\end{itemize}
Such an expansion of a trapping set can be performed multiple times by adding one neighboring variable node of degree 2, each time. This is summarized in Algorithm 2. In a general case, Algorithm 2 can be used with Algorithm 1 to expand the trapping sets found by Algorithm 1. This is summarized in Algorithm 3.

\hspace{-.25in}
\linethickness{0.275mm}
\line(1,0){250} \\
\textbf{Algorithm 2}: Finding trapping sets of size up to $k$ with the number of unsatisfied check nodes up to $T$ constructed by adding degree-2 variable nodes to the input trapping sets for an irregular LDPC code with the Tanner graph $G=(L\cup R\,,E)\,$.\\
($\mathcal{L}_{in}$ and $\mathcal{L}_{out}$ are the lists of input and output trapping sets, respectively)\\
\linethickness{0.125mm}
\line(1,0){250}
\begin{algorithmic}[1]
\STATE  \textbf{Inputs:} $G$, $\mathcal{L}_{in}$
\STATE $\mathcal{L}_{out}$ $\leftarrow \emptyset$.
\REPEAT
\STATE Select an element of $\mathcal{L}_{in}$ with size less than $k$, and denote it as  $t$.
\STATE Form the set $N_2(t)$ which contains variable nodes of degree 2 in $L\backslash t$ that are connected to at least one unsatisfied check node of $t$, i.e., to $\Gamma_{\mathrm{o}}(t)$.
\label{step1}
\FOR {each node $v$ in $N_2(t)$}
\STATE $t'\leftarrow t \cup \{v\}$.
\IF {($t' \in \mathcal{T}$) \footnote{This condition ensures that each variable node of degree larger than 2 is connected to at least 2 satisfied check nodes. The condition has no bearing on  degree-2 variable nodes.} and ($t' \notin \mathcal{L}_{out}$) and ($|\Gamma_{\mathrm{o}}(t')|\le T$) }
\label{step2}
\STATE $\mathcal{L}_{out}\leftarrow \mathcal{L}_{out} \cup \{t'\}$.
\ENDIF
\ENDFOR
\UNTIL { all the elements of $\mathcal{L}_{in}$ are selected.}
\STATE  \textbf{Output:} $\mathcal{L}_{out}$.
\end{algorithmic}

\begin{remark}
Note that in Algorithm 2, the number of unsatisfied check nodes of the resultant trapping sets never increases. Hence, to find trapping sets of size $a$ with less than $b$ unsatisfied check nodes, one should consider all the ($a',b')$ trapping sets with $a'<a,\,b'<b.$\footnote{Although this condition may not cover all the trapping sets discussed in Part $c$ of Section IV.B, our simulations show that for the tested codes, almost all the trapping sets are in fact found by Algorithm 2. The trapping sets that are missed by Algorithm 2 are the ones that can {\em only} be obtained by starting from trapping sets with larger number of unsatisfied check nodes.} It should be mentioned that since every single variable node of degree $d_v$ can be regarded as a $(1,d_v)$ trapping set, to find the trapping sets with less than $b$ unsatisfied check nodes, we consider also all the variable nodes of degree $d_v\le b$ as part of the initial set. For example, for the case of $b=3$, starting with a single variable node of degree $d_v=2$ or $d_v=3$, two typical structures of the resultant trapping sets are shown in Figs.~\ref{degree2_3}(a) and~\ref{degree2_3}(b), respectively. Note that starting from a degree-2 variable node and performing the above steps results in finding a 2-chain.
\end{remark}
%

\begin{figure}[!t]
\centering
\includegraphics[width=1.6in]{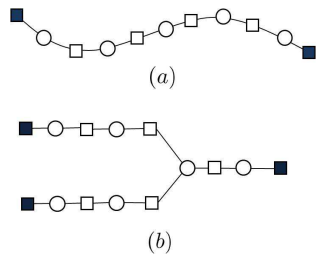}
\caption{Typical expansions of degree-2 and degree-3 variable nodes by adding the neighboring degree-2 variable nodes.}
\label{degree2_3}
\end{figure}

\hspace{-.25in}
\linethickness{0.275mm}
\line(1,0){250} \\
\textbf{Algorithm 3}: Finding trapping sets of size up to $k$ with the number of unsatisfied check nodes up to $T$ for an irregular LDPC code with the Tanner graph $G=(L\cup R\,,E)\,$.\\
($\mathcal{L}_{in}$ and $\mathcal{L}_{out}$ are the lists of input and output trapping sets, respectively.)\\
\linethickness{0.125mm}
\line(1,0){250}
\begin{algorithmic}[1]
\STATE  \textbf{Inputs:} $G$, $\mathcal{L}_{in}$
\STATE Use $\mathcal{L}_{in}$ as the input of Algorithm 1
\STATE $\mathcal{L}1_{out}=$ trapping sets found by Algorithm 1
\STATE $\mathcal{L}2_{in}= \mathcal{L}1_{out}\bigcup$ $\{$low degree variable nodes$\}$
\STATE Use $\mathcal{L}2_{in}$ as the input of Algorithm 2
\STATE $\mathcal{L}_{out}=$ trapping sets found by Algorithm 2
\STATE \textbf{Output:} $\mathcal{L}_{out}$.
\end{algorithmic}
\linethickness{0.275mm}
\line(1,0){250}

\begin{remark}
For irregular codes, in addition to short cycles, cycles with low ACE are also considered as part of the initial input set of Algorithm 1. This is because these cycles may not be found using the expansion process of Algorithm 1. Algorithm 1 finds the smallest trapping sets containing the input, which are usually the combination of the input and a short cycle (or a structure described in Lemmas 2 -- 4). Since variable nodes of large degree are more likely to be part of such structures, the outputs of Algorithm 1 are usually the combinations of the input and variable node(s) of large degree. This is while cycles with low ACE are generally constructed by low degree variable nodes. Cycles with low ACE can be easily found by monitoring the ACE value during the execution of a cycle finding algorithm.
\end{remark}
\begin{remark}
As an alternative approach to using Algorithm 3, one can only use Algorithm 2 with the
variable nodes of low degree and cycles with low ACE as the initial input set, and then recursively expand them to larger trapping sets. It should however be noted that for the irregular LDPC codes with a small fraction of degree-2 variable nodes, this approach may not find all the dominant trapping sets of the code.
\end{remark}

\section{Numerical Results}
For the simulations, we assume binary phase-shift keying (BPSK) modulation over the AWGN channel with coherent detection. Notations $E_b$ and $N_0$ are used for the average energy per information bit and the one-sided power spectral density of the AWGN, respectively.
\subsection{Regular Codes}
We have applied the proposed algorithm successfully to a large number of regular LDPC codes.
Here, we only present the results for four of them. The first three examples are random and structured LDPC codes whose dominant trapping sets have already been reported in the literature and thus provide us with a reference for comparison. The fourth example is a random LDPC code of rate $1/2$ with variable node degree 4. To verify the trapping sets found by the proposed algorithm for this code, we estimate the error floor using importance sampling \cite{Cole2008} based on the obtained trapping sets and demonstrate that the estimation is practically identical to the results of Monte Carlo simulations.
The reported running times in the following examples are for a desktop computer with $2$-GHz CPU and $1$ GB of RAM.

\begin{example}
We consider an LDPC code constructed by the progressive edge growth (PEG) algorithm \cite{Hu2005} $($\emph{PEGReg252x504} of~\cite{Online}$)$. This code is left-regular with the left degree 3, and girth 8. The same code was also investigated in \cite{Kyung2010} and the distribution of its fully absorbing sets was determined. For Algorithm 1, the short cycles of length $g$, $g+2$ and $g+4$  were used as the initial input set. The algorithm was limited to
finding trapping sets of maximum size 13, and the threshold $T$ was selected such that only
the trapping sets with the two smallest values of $b$ for each size were considered. (Using a larger $T$ has no effect on the accuracy of the results reported here.) Since all the variable nodes have degree 3, all the trapping sets found by Algorithm 1 are absorbing sets. Fully absorbing sets were found by examining the obtained absorbing sets and testing them for the definition of a fully absorbing set. Table~\ref{pegreg} shows the absorbing sets and the fully absorbing sets found by the proposed algorithm and their multiplicities. In the table, we have also reported the results obtained by the exhaustive search algorithm of \cite{Kyung2010}, for comparison. (Note that the hyphen notation ``-" in the table means that no data was reported.)
As can be seen from Table \ref{pegreg}, for many classes of trapping sets, the proposed algorithm found exactly the same number of fully absorbing sets as the exhaustive search algorithm of \cite{Kyung2010} did. For the other classes, the difference between the two sets of results is rather small. Moreover, the proposed algorithm found $(11,3)$, $(13,3)$, $(10,4)$ and $(12,4)$  fully absorbing sets which are out of the reach of the exhaustive search algorithm. It is worth mentioning that the exhaustive search algorithm of \cite{Kyung2010} took about 7 hours to find only the first three rows of Table \ref{pegreg} \cite{Kyung2010} (needless to say, the larger the size of the absorbing sets, the longer the running time of the algorithm). This is while Algorithm 1 took only 10 minutes to find all the absorbing sets listed in Table \ref{pegreg}.

\begin{table}[!h]
\caption{\small{Dominant Absorbing Sets(ABS) and Fully Absorbing Sets of the \emph{PEGReg252x504} Code Obtained by the Proposed Algorithm and the Exhaustive Search Algorithm of \cite{Kyung2010}}}
\centering
\begin{tabular}{|c|c|c|c|c|}
\hline
\cline{1-4}
\small{\textbf{Trapping}}&\small{\textbf{Proposed}} &\small{\textbf{Proposed}}&\small{\textbf{Exhaustive}} \\
\small{\textbf{Set}} &\small{\textbf{Algorithm}}&\small{\textbf{Algorithm}}&\small{\textbf{Search \cite{Kyung2010}}} \\
\small{\textbf{}} &\small{\textbf{(ABS)}}&\small{\textbf{(Fully ABS)}}&\small{\textbf{(Fully ABS)}} \\

\hline\hline
\small{ {(4, 4)}}&\small{802}&\small{760}&\small{760}\\\hline
\small{ {(5, 3)}}&\small{14}&\small{14}&\small{14}\\\hline
\small{(5, 5)}&\small{11279}&\small{10156}&\small{10156}\\\hline
\small{ {(6, 4)}}&\small{985}&\small{849}&\small{849}\\\hline
\small{(6, 6)}&\small{86391}&\small{66352}&\small{66352}\\\hline
\small{ {(7, 3)}}&\small{57}&\small{47}&\small{47}\\\hline
\small{(7, 5)}&\small{27176}&\small{{21810}}&\small{{22430}}\\\hline
\small{ {(8, 2)}}&\small{5}&\small{4}&\small{4}\\\hline
\small{(8, 4)}&\small{2610}&\small{{2258}}&\small{{2270}}\\\hline
\small{ {(9, 1)}}&\small{1}&\small{1}&\small{1}\\\hline
\small{(9, 3)}&\small{156}&\small{146}&\small{146}\\\hline
\small{ {(10, 2)}}&\small{6}&\small{6}&\small{6}\\\hline
\small{(10, 4)}&\small{7929}&\small{{6691}}&\small{{-}}\\\hline
\small{ {(11, 3)}}&\small{605}&\small{{558}}&\small{{-}}\\\hline
\small{ {(12, 2)}}&\small{25}&\small{{24}}&\small{{26}}\\\hline
\small{(12, 4)}&\small{23668}&\small{{19959}}&\small{{-}}\\\hline
\small{ {(13, 1)}}&\small{1}&\small{1}&\small{1}\\\hline
\small{(13, 3)}&\small{2124}&\small{{1954}}&\small{{-}}\\\hline

\cline{1-4}
\end{tabular}
\label{pegreg}
\end{table}
\end{example}
\begin{example}
In this example, we consider the Tanner $(155,\,64)$ code~\cite{Tanner-01}.
This code was also investigated in~\cite{Wang2009}. The exhaustive search algorithm of \cite{Wang2009}
showed that this code has no trapping set of length less than $8$ with $2$ unsatisfied check nodes and has no
trapping set of length up to $11$ with $1$ unsatisfied check node. It was also shown in \cite{Wang2009}
that the code has $465$ $(8,2)$ trapping sets.

The girth for the Tanner graph of this code is $g=8$. The short cycles of length $g$, $g+2$ and $g+4$ were
used as the initial inputs to Algorithm 1. The algorithm was limited to only
find trapping sets of maximum size 12 and the threshold $T$ was selected such that only the trapping sets with the two smallest values of $b$ for each size were considered. Table~\ref{tanner155} shows the trapping sets found by the proposed algorithm
and their multiplicity. As can be seen in the table, the algorithm found all the 465 $(8,2)$ trapping sets among others. All the trapping sets in Table \ref{tanner155} were found
in less than $2$ minutes.

To further verify that the obtained trapping sets do
in fact include the dominant ones, we performed Monte Carlo simulations on the code
with a 4-bit quantized min-sum decoder over the AWGN channel at
signal-to-noise ratio (SNR) of 6.5 dB (which is in the error floor region of this code).
Among the $300$ error patterns, about $90\%$ were $(8,2)$ trapping sets,
about $8\%$ were $(10,2)$ trapping sets, and only 2 did not belong to the sets reported in
Table~\ref{tanner155}.
\label{ex1}
\end{example}

\begin{table}[!h]
\caption{\small{Dominant Trapping Sets of the Tanner $(155,\,64)$ Code Obtained by the Proposed Algorithm}}
\centering
\begin{tabular}{|c|c|c|c|c|}
\hline
\cline{1-2}
\small{\small{\textbf{Trapping Set}}}&\small{\textbf{Multiplicity}}\\
\hline\hline
\small{(4,4)}&\small{465}\\\hline
\small{(5,3)}&\small{155}\\\hline
\small{(6,4)}&\small{930}\\\hline
\small{(7,3)}&\small{930}\\\hline
\small{(8,2)}&\small{465}\\\hline
\small{(9,3)}&\small{1395}\\\hline
\small{(10,2)}&\small{1395}\\\hline
\small{(11,3)}&\small{1860}\\\hline
\small{(12,2)}&\small{930}\\\hline

\cline{1-2}
\end{tabular}
\label{tanner155}
\end{table}

\begin{example}
As the third example, we consider the Margulis $(2640,\,1320)$ code~\cite{Margulis},~\cite{Online}.
It is known that the most dominant trapping sets of this code are $1320$ $(12,4)$
and $1320$ $(14,4)$ trapping sets~\cite{Richardson2003}. The Tanner graph of this code
has girth $g = 8$. The set of short cycles of
length $g$, $g+2$ and $g+4$ was used as the
input set of the proposed algorithm. The algorithm was limited to use only the trapping sets with the two smallest values of $b$ for each size. Since the degree of all the variable nodes of this code is 3, all the trapping sets found by Algorithm 1 are also absorbing sets. The first column in Table~\ref{marg2640} shows the dominant
absorbing sets found by Algorithm 1. For comparison, the dominant trapping sets obtained by the algorithm of \cite{Abu2010} are listed in the last column of Table~\ref{marg2640}. It should be noted that in \cite{Abu2010} there is no condition on the number of satisfied check nodes connected to each variable node. Thus to have a fair comparison, we also consider the trapping sets constructed by the combination of trapping sets found by Algorithm 1 and one of their neighboring variable nodes. The second column of Table~\ref{marg2640} shows the number of such trapping sets.\footnote{Our simulations indicate that the effect of extra trapping sets found by removing the constraint on the number of satisfied check nodes connected to each variable node of the trapping set on the error floor performance of the code is negligible.} As can be seen, for all the trapping set classes, the proposed algorithm performs at least as well as the algorithm of \cite{Abu2010}. Moreover, the required time for the algorithm of \cite{Abu2010} was 7 days on a 2.8 GHz PC \cite{Abu2010}, while the proposed algorithm took about 5 hours to finish. As another comparison for the running time of the proposed algorithm, it took the algorithm 55 minutes to find all the absorbing sets of size less than 15, while the same task took 8.2 hours  for the impulse method of \cite{Cole2008} on a comparable
computer (2.2-GHz CPU with 1 GB RAM).

\label{ex2}
\end{example}

\begin{table}[!h]
\caption{\small{Dominant Trapping Sets of the Margulis $(2640,\,1320)$ Code Obtained by the Proposed Algorithm and the Algorithm of \cite{Abu2010}}}
\centering
\begin{tabular}{|l|c||l|c|c|c|c|}
\hline
\cline{1-4}
\small{\textbf{Trapping}} &\small{\textbf{Proposed}}&\small{\textbf{Proposed}}&\small{\textbf{Algorithm}} \\
\small{\textbf{Set}} &\small{\textbf{Algorithm}}&\small{\textbf{Algorithm}}&\small{\textbf{ of \cite{Abu2010}}} \\
\small{\textbf{}} &\small{\textbf{(Absorbing)}}&\small{\textbf{(Trapping)}}&\small{\textbf{(Trapping)}} \\
\hline\hline
\small{(7, 5)}&\small{7920}&\small{7920}&\small{-}\\\hline
\small{(8, 6)}&\small{106920}&\small{$>$106920}&\small{-}\\\hline
\small{(9, 5)}&\small{2640}&\small{2640}&\small{-}\\\hline
\small{(10, 6)}&\small{117480}&\small{$>$117480}&\small{-}\\\hline
\small{(11, 5)}&\small{5280}&\small{5280}&\small{9}\\\hline
\small{(12, 4)}&\small{1320}&\small{1320}&\small{1320}\\\hline
\small{(13, 5)}&\small{2640}&\small{26400}&\small{2699}\\\hline
\small{(14, 4)}&\small{1320}&\small{1320}&\small{1320}\\\hline
\small{(15, 5)}&\small{0}&\small{26400}&\small{7938}\\\hline
\small{(16, 6)}&\small{0}&\small{258347}&\small{21153}\\\hline
\small{(17, 5)}&\small{5280}&\small{5280}&\small{0}\\\hline
\small{(18, 6)}&\small{0}&\small{132000}&\small{2642}\\\hline

\cline{1-4}
\end{tabular}
\label{marg2640}
\end{table}

\begin{example}
For this example, we consider a $(1008,\,504)$ random code with variable node degree 4 and check node degree 8 constructed by the program of \cite{Online}.\footnote{Using \emph{code6.c} with seed=380.} This code has one cycle of length 4 ($C_4$). In addition to that, the short cycles of length 6 to 10 were used as the initial input set for Algorithm 1. The algorithm was constrained to find trapping sets of size up to 12 and to use only the trapping sets with the two smallest values of $b$ for each size. Table~\ref{R4_8} shows the dominant trapping sets found by Algorithm 1 and their multiplicities. It is worth mentioning that none of the trapping sets listed in Table \ref{R4_8} contains any of the variable nodes participating in $C_4$. The trapping sets reported in Table~\ref{R4_8} were used to estimate the error floor of the code using the importance sampling technique described in \cite{Cole2008}. Fig.~\ref{IS_MC_3} shows the Monte Carlo simulation results for the frame error rate (FER) and the corresponding error floor estimation based on importance sampling. The results are for a 3-bit min-sum decoder with a maximum number of 50 iterations. As can be seen in Fig.~\ref{IS_MC_3}, the estimation closely matches the Monte Carlo simulation, verifying the dominance of the trapping sets found by Algorithm 1. Monte Carlo simulations also revealed that the most harmful trapping set of this code is the $(6,\,4)$ trapping set. In fact, in almost all the decoding failures, the decoder converged to the (6,4) trapping set. As can be seen in Table~\ref{R4_8}, all the trapping sets have at least 4 unsatisfied check nodes. This makes the exhaustive search methods of ~\cite{Wang2007}, \cite{Wang2009}, \cite{Kyung2010} ineffective for finding the dominant trapping sets of this code. This is while all the trapping sets in Table \ref{R4_8} were found in less than 5 minutes by the proposed algorithm.
\end{example}
\begin{table}[!h]
\caption{\small{Dominant Trapping Sets of the $(1008,\,504)$ Regular LDPC Code ($d_v=4,\,d_c=8$) Obtained by the Proposed Algorithm}}
\centering
\begin{tabular}{|c|c|c|c|c|}
\hline
\cline{1-2}
\small{\small{\textbf{Trapping Set}}}&\small{\textbf{Multiplicity}}\\
\hline\hline
\small{(5,6)}&\small{15}\\\hline
\small{(6,4)}&\small{1}\\\hline
\small{(6,6)}&\small{36}\\\hline
\small{(7,5)}&\small{13}\\\hline
\small{(8,6)}&\small{5}\\\hline
\small{(9,6)}&\small{5}\\\hline
\small{(10,6)}&\small{3}\\\hline
\small{(11,6)}&\small{3}\\\hline
\small{(12,8)}&\small{75}\\\hline

\cline{1-2}
\end{tabular}
\label{R4_8}
\end{table}

\begin{figure}[!t]
\centering
\includegraphics[width=5.8in]{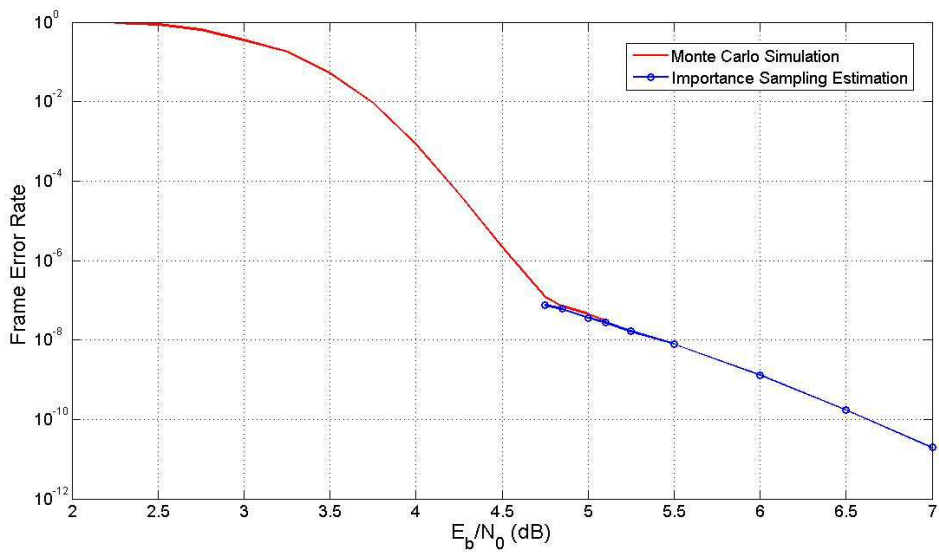}
\caption{Error floor estimation and Monte Carlo simulation for the $(1008,\,504)$ regular LDPC code ($d_v=4,\,d_c=8$).}
\label{IS_MC_3}
\end{figure}

\subsection{Irregular Codes}
In this section, we present the results of applying the proposed algorithm to three irregular LDPC codes. To find the dominant trapping sets of the irregular codes, we used two approaches. In the first approach, we used Algorithms 1 and 2 in the framework described in Algorithm 3. In this approach, as the first step, we used the short cycles of the codes, as well as the low ACE cycles as the initial input set, and applied Algorithm 1. We then used the trapping sets found by Algorithm 1 along with the variable nodes of low degree, and applied Algorithm 2 to expand them. As the second approach, we only used the variable nodes of low degree and cycles with low ACE as the initial input set, and then used Algorithm 2 to recursively expand them to larger trapping sets. Interestingly, for all three codes, the results of the second approach were very close to those of the first one.
\begin{example}
For this example, we consider the irregular LDPC code constructed by the PEG algorithm (\emph{PEGirReg252x504} code~\cite{Online}). This code was also investigated in \cite{Kyung2010} for its fully absorbing sets. For Algorithm 1, the short cycles of length $g$, $g+2$, and  the cycles with length less than 20 and ACE less than 4 were used as the initial input set. The algorithm was constrained to
find only trapping sets of size less than 12 and the threshold $T$ was selected such that only the trapping sets with the four smallest values of $b$ for each size were considered. The resultant trapping sets and variable nodes of degree 2 and 3 were then expanded by adding neighboring degree-2 variable nodes, and finally were examined to find the fully absorbing sets. Table~\ref{pegir} shows the fully absorbing sets found by Algorithm 3 and the exhaustive search algorithm of \cite{Kyung2010}. It should be noted that, similar to \cite{Kyung2010}, we relaxed the condition that degree-2 variable nodes of (fully) absorbing sets must be connected to two satisfied check nodes.
As can be seen from Table \ref{pegir}, the proposed algorithm found almost all the fully absorbing sets of this code.\footnote{The multiplicity for trapping sets  $(7,2)$ and $(8,2)$ are reported as 274 and 468 in \cite{Kyung2010}, respectively. Moreover, no $(7,1)$ or $(8,1)$ trapping set is reported in \cite{Kyung2010}. The values reported for these four trapping sets in the last column of Table \ref{pegir} are based on \cite{Kyung_p}.} Moreover, the proposed algorithm found a number of $(a,\,1)$ trapping sets for $a\ge9$, which were not reported in \cite{Kyung2010}. For the second approach, the cycles of length up to 20 with ACE lower than 4 and the variable nodes of degree 2 and 3 were used as the initial inputs, and the algorithm found almost the same trapping sets as in the first approach. For the running time, the first and the second approaches took 15 minutes and 5 minutes, respectively.
\end{example}
\begin{table}[!h]
\caption{\small{Dominant Fully Absorbing Sets of the \emph{PEGirReg252x504} Code Obtained by the Proposed Algorithm and the Algorithm of \cite{Kyung2010}}}
\centering
\begin{tabular}{|c|c|c|c|c|}
\hline
\cline{1-3}
\small{\textbf{Trapping}} &\small{\textbf{Proposed}}&\small{\textbf{Exhaustive}} \\
\small{\textbf{Set}} &\small{\textbf{Algorithm}}&\small{\textbf{Search \cite{Kyung2010}}} \\
\hline\hline
\small{{(3, 2)}}&\small{219}&\small{219}\\\hline
\small{{(4, 2)}}&\small{208}&\small{208}\\\hline
\small{(5, 2)}&\small{198}&\small{198}\\\hline
\small{{(6, 2)}}&\small{205}&\small{205}\\\hline
\small{(7, 1)}&\small{{ 2}}&\small{{  2}}\\\hline
\small{{(7, 2)}}&\small{{ 271}}&\small{{ 272}}\\\hline
\small{(8, 1)}&\small{{ 8}}&\small{{  8}}\\\hline
\small{{(8, 2)}}&\small{{ 458}}&\small{{ 460}}\\\hline
\small{(9, 1)}&\small{{ 16}}&\small{{  -}}\\\hline
\small{{(9, 2)}}&\small{{ 855}}&\small{{  -}}\\\hline
\small{(10, 1)}&\small{{ 22}}&\small{{  -}}\\\hline
\small{{(10, 2)}}&\small{{ 1533}}&\small{{  -}}\\\hline
\small{(11, 1)}&\small{{ 36}}&\small{{  -}}\\\hline

\cline{1-3}
\end{tabular}
\label{pegir}
\end{table}
\begin{example}
 For this example, we used the $(1944,972)$ structured irregular code with rate $1/2$, adopted in the IEEE 802.11 standard \cite{IEEE802.11n}. We used the same parameters as in the previous example for the two approaches. Table~\ref{1944_972} shows the number of dominant trapping sets of different sizes found by the algorithm of \cite{Abu2010} and the proposed approaches. For this code, both of our approaches found exactly the same set of trapping sets. In fact, all the trapping sets listed in Table~\ref{1944_972} have one of the following three structures: a 2-chain, a single cycle with low ACE, and the combination of a 2-chain and a single cycle of low ACE. For example, all the trapping sets of size less than 7 listed in Table~\ref{1944_972} are 2-chains, and all the $(12,1)$ trapping sets are single cycles of eleven degree-2 variable nodes and one degree-3 variable node.  As can be seen in Table~\ref{1944_972}, for all classes of trapping sets, the proposed algorithms found at least as many trapping sets as the algorithm of \cite{Abu2010} did. The first and the second approaches took 45 and 5 minutes, respectively, to find all the trapping sets in Table~\ref{1944_972}. This is while the algorithm of \cite{Abu2010} took 5 days (on a 2.8-GHz CPU) to find the results reported in Table~\ref{1944_972}.

\begin{table}[!h]
\caption{\small{Dominant trapping sets of the (1944,972) code obtained by the proposed algorithm}}
\centering
\begin{tabular}{|c|c|c|c|c|}
\hline
\cline{1-3}
\small{\textbf{Trapping}} &\small{\textbf{Proposed}}&\small{\textbf{Algorithm}} \\
\small{\textbf{Set}} &\small{\textbf{Algorithm}}&\small{\textbf{of \cite{Abu2010}}} \\
\hline\hline
\small{{(2, 2)}}&\small{810}&\small{-}\\\hline
\small{{(3, 2)}}&\small{729}&\small{-}\\\hline
\small{{(4, 2)}}&\small{648}&\small{648}\\\hline
\small{{(5, 2)}}&\small{567}&\small{567}\\\hline
\small{(6, 2)}&\small{486}&\small{486}\\\hline
\small{{(7, 2)}}&\small{{486}}&\small{{485}}\\\hline
\small{(8, 2)}&\small{{648}}&\small{{637}}\\\hline
\small{{(9, 2)}}&\small{{972}}&\small{{-}}\\\hline
\small{(10, 2)}&\small{{1377}}&\small{{1210}}\\\hline
\small{{(11, 2)}}&\small{{1944}}&\small{{1635}}\\\hline
\small{(12, 1)}&\small{81}&\small{81}\\\hline
\small{{(12, 2)}}&\small{{2754}}&\small{{2166}}\\\hline
\small{(13, 1)}&\small{162}&\small{162}\\\hline
\small{{(14, 1)}}&\small{162}&\small{162}\\\hline
\small{(15, 1)}&\small{{162}}&\small{{-}}\\\hline
\small{{(16, 1)}}&\small{{162}}&\small{{-}}\\\hline
\small{{(17, 1)}}&\small{{162}}&\small{{-}}\\\hline
\small{(18, 1)}&\small{{81}}&\small{{-}}\\\hline

\cline{1-3}
\end{tabular}
\label{1944_972}
\end{table}

 Based on the importance sampling technique of \cite{Cole2008}, the trapping sets in Table \ref{1944_972} with size $l$, $6\le l\le 12$, were used to estimate the error floor of this code for a 3-bit quantized min-sum decoder over the AWGN channel. Fig.~\ref{MC_IS_1944} shows the error floor estimation and the Monte Carlo simulation results for this code. As can be seen in Fig.~\ref{MC_IS_1944}, the importance sampling estimation closely matches the Monte Carlo simulation, further verifying the dominance of the trapping sets found by the proposed algorithm.

\end{example}
\begin{figure}[!h]
\centering
\includegraphics[width=5.8in]{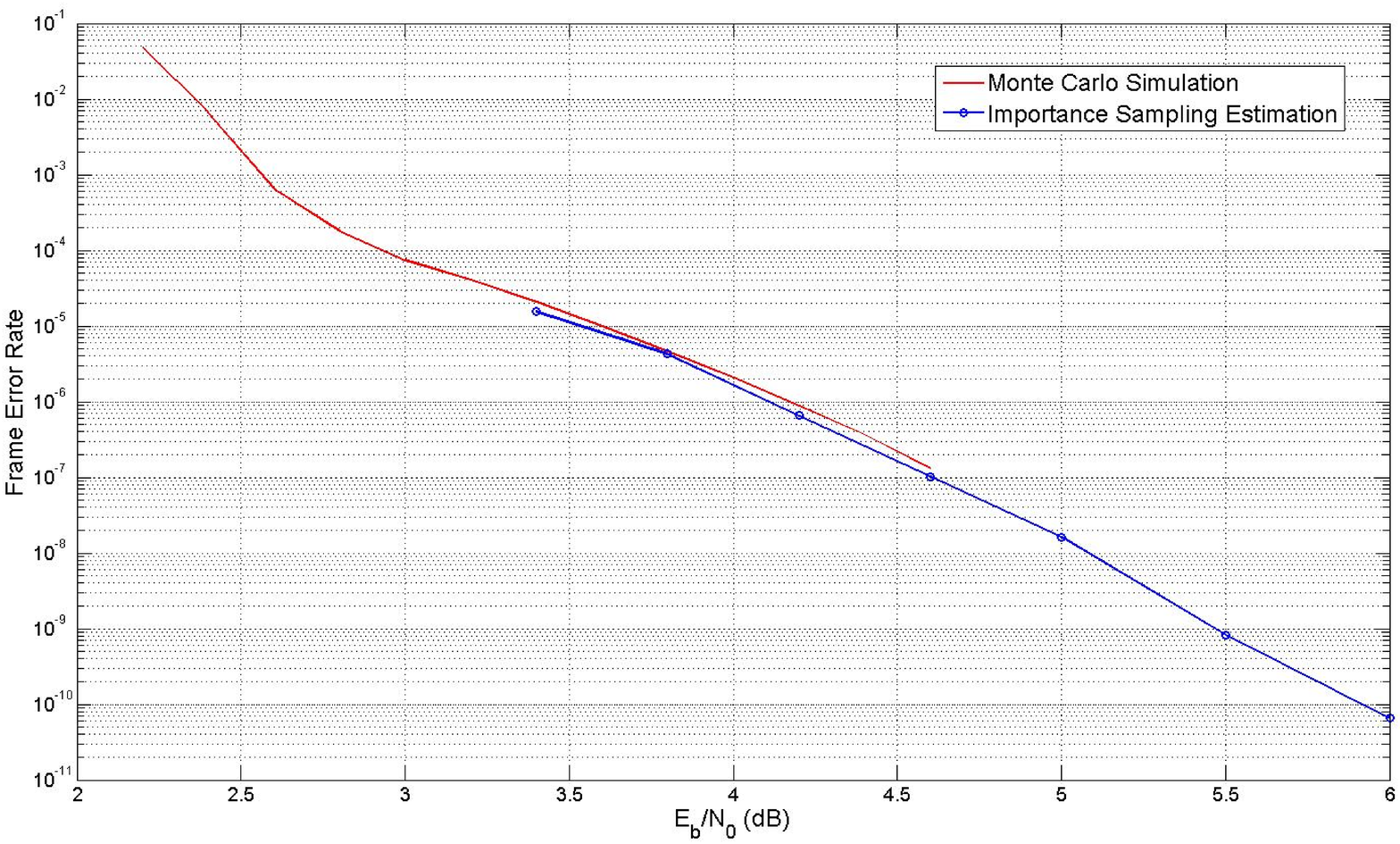}
\caption{Error floor estimation and Monte Carlo simulation for the $(1944,972)$ irregular LDPC code. }
\label{MC_IS_1944}
\end{figure}

\begin{example}
As the last example, we use the following degree distribution optimized for the min-sum algorithm in \cite {Chung2000} and construct a $(1000,\,499)$ LDPC code using the PEG algorithm: $\lambda(x) = .30370x+.27754x^2+.02843x^5+.20014x^6+.19019x^{19}$ and $ \rho(x) = .0160x^5 + .9840x^6$. The girth of the resultant graph is 6, and we use the short cycles of length 6 and 8, and cycles of length up to 20 with ACE less than 4 as the initial input set of Algorithm 2. It takes 1 minute for the algorithm to find the trapping sets of size up to 10. Based on the obtained trapping sets and using the importance sampling, we estimate the error floor of the code. Fig.~\ref{MC_IS_irpeg7} shows the estimation and Monte Carlo simulations for this code. As can be seen in this figure, the estimation closely matches the Monte Carlo simulation results, verifying that the dominant trapping sets of the code have been found by the algorithm.

\begin{figure}[!t]
\centering
\includegraphics[width=5.8in]{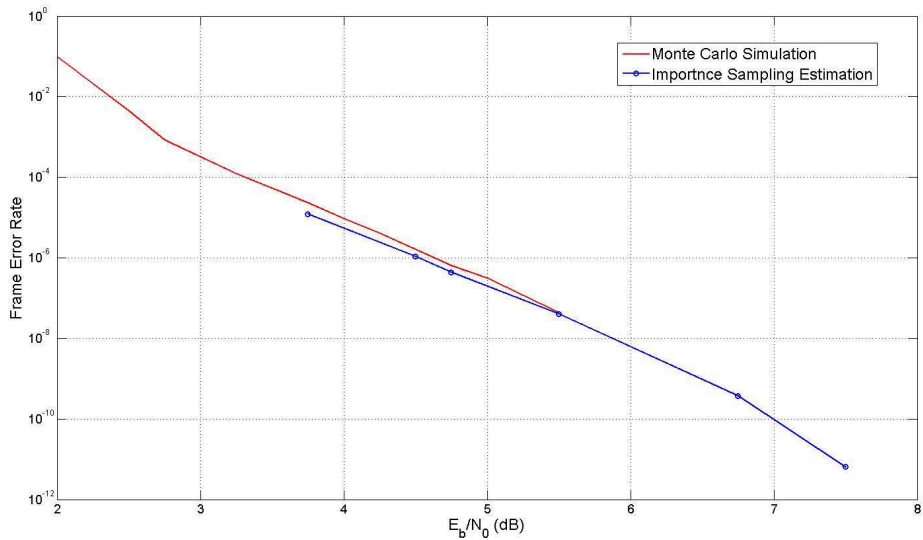}
\caption{Error floor estimation and Monte Carlo simulation for the $(1000,\,499)$ irregular LDPC code.}
\label{MC_IS_irpeg7}
\end{figure}

\end{example}

\section{Conclusions}

In this paper, we proposed an efficient algorithm for finding the dominant trapping sets of an LDPC code.
The algorithm starts from an initial set of trapping sets and recursively and greedily expands them to trapping
sets of larger size. The initial set for regular codes is a set of short cycles, and for irregular codes,
it also includes variable nodes of small degree and cycles with low ACE values. To devise the expansions,
the structure of dominant trapping sets is carefully studied for both regular and irregular codes. The efficiency
and accuracy of the proposed algorithm was demonstrated through a number of examples. It was observed that
the proposed algorithm is faster by up to about two orders of magnitude compared to similar search algorithms.

\appendices
\section{}
In this appendix, we present Lammas ~\ref{lem_nonel} and ~\ref{2chain_size}, used in Sections III and IV, respectively, along with their proofs.
The appendix also contains the proofs for Lemmas \ref{lem99} and \ref{lem10}(i).

\begin{lemma}
For a left-regular graph $G$ with left degree $d_l\ge3$ and girth $g>4$,\footnote{For the case of $d_l = 2$, it is easy to see that any $(a,b)$ elementary trapping set has $b=0$ or $b=2$. For $b=0$, the smallest value of $a$ is $g/2$, which corresponds to the trapping set being a shortest cycle. For an elementary trapping set with $b=2$, the smallest value of $a$ is one, which corresponds to a single variable node. For a non-elementary $(a,b)$ trapping set however, if $b=0$, the smallest value of $a$ is $g$. If $b=2$, the minimum value of $a$ for such a trapping set is $3$.} consider an $(a,\,b)$ trapping set with $b< a$. If such a trapping set is elementary, let the notation $a_e$ denote its size, and consider the case where $d_l (d_l - 1) > b$. Otherwise, for non-elementary trapping sets with $b < a$, let the notations $a_{n1}$ and $a_{n2}$ denote the size of the trapping set if it has at least one unsatisfied check node of degree $d_o >1$ and one satisfied check node of degree $d_e > 2$ in $G(\cal{S})$, respectively. For the two latter cases, suppose that $d_o (d_l - 1) > b$ and $d_e (d_l - 1) > b$, respectively. Then depending on the value of $g$, we have the following two sets of inequalities:
\begin{itemize}
\item [a)] For $g=4k$, where $k$ is an integer larger than 1, we have:\\
\[a_e     \ge 1+d_l+(d_l(d_l-1)-b)\sum_{i=0}^{k-3}{(d_l-1)^i}+ \frac{(d_l(d_l-1)-b)(d_l-1)^{k-2}}{d_l}\,\,,        \]
\[a_{n1} \ge d_e+(d_e(d_l-1)-b)\sum_{i=0}^{k-2}{(d_l-1)^i}\,\,;    \]
\[a_{n2} \ge d_o+(d_o(d_l-1)-b+1)\sum_{i=0}^{k-2}{(d_l-1)^i}\,.    \]

\item [b)] For $g=4k+2$, where $k$ is a positive integer, we have:\\
\[a_e     \ge 1+d_l+(d_l(d_l-1)-b)\sum_{i=0}^{k-2}{(d_l-1)^i}\,\, ,       \]
\[a_{n1} \ge d_e+(d_e(d_l-1)-b)\sum_{i=0}^{k-2}{(d_l-1)^i}+ \frac{(d_e(d_l-1)-b)(d_l-1)^{k-1}}{d_l}\,\,;    \]
\[a_{n2} \ge d_o+(d_o(d_l-1)-b+1)\sum_{i=0}^{k-2}{(d_l-1)^i}+ \frac{(d_o(d_l-1)-b+1)(d_l-1)^{k-1}}{d_l}\, .   \]

\end{itemize}
\label{lem_nonel}
\end{lemma}

\begin{proof}
Here, we just present the sketch of the proof. For this, we first need the following lemma, whose proof follows later in the appendix.

\emph{Lemma} \ref{lem10}(i): In a left-regular graph $G$ with left degree $d_l\ge2$, if the induced subgraph $G({\cal S})$ of an $(a,\,b)$ trapping set $\mathcal{S} $ does not contain any cycle, then $b\ge a(d_l-2)+2$. The inequality is satisfied with equality for elementary trapping sets.

Based on Lemma \ref{lem10}(i), it is clear that a trapping set with $b < a$ has at least one cycle. Therefore, considering any variable (or check) node of $\mathcal{S} $ as the root, and growing $G(\mathcal{S})$ from that node, one can construct a tree of at least $g/2$ layers, where the layers contain either variable or check nodes alternately, with no repetition of nodes. The number of variable nodes in this tree can be used as a lower bound on the number of variable nodes in $\mathcal{S} $. In this tree, the number of check nodes in layer $i>1$ of the tree, $N_c^i$, is $N_c^i=(d_l-1)N_v^{i-1}$, where $N_v^{i-1}$ is the number of variable nodes in layer $i-1$. Similarly, $N_v^{i}=\sum{(d_{c_j^{i-1}}-1)}$, where $d_{c^{i-1}_j}$ is the degree (within $G({\cal S})$) of the $j^{th}$ check node in layer $i-1$, and the summation is over all the check nodes in layer $i-1$. To minimize the number of variable nodes in the tree, one needs to make $\sum{(d_{c_j^{i-1}}-1)}$ as small as possible in each check node layer of the tree. In particular, this should be done at the upper layers of the tree if possible, since these layers contribute the most in the total number of variable nodes in the tree. In addition, to obtain a lower bound on the size of the trapping sets, we assume that even for the non-elementary case, except for one check node, the degrees of all the other check nodes in $G(\mathcal{S})$ are either 1 or 2. Moreover, we assume that all the check nodes of degree 1 are in the first (upper) layer(s) of check nodes after the root layer.

For the case of an elementary trapping set, according to the assumption of $b< a$, there is at least one variable node that is not connected to any unsatisfied check nodes. Considering such a variable node as the root node, all the check nodes in the first layer are satisfied check nodes. That is, $N_v^0=1$ (root node), $N_c^1=N_v^2=d_l$, $N_c^3=d_l(d_l-1)$, $N_v^4=d_l(d_l-1)-b$ and  $N_c^{i-1}=(d_l-1)N_v^{i-2}$, $N_v^i=N_c^{i-1}$, for $i=6,8,\ldots$.\footnote{Here, based on the statement of the lemma, we have assumed that all the unsatisfied check nodes can fit in the third layer of the tree. In the case that $d_l(d_l-1)-b \le 0$, some of the unsatisfied check nodes have to be located in the next layer(s), and the above equations and the claims of the lemma will have to be accordingly revised.} Therefore, the total number of variable nodes in the constructed tree is $1+d_l+(d_l(d_l-1)-b)+(d_l(d_l-1)-b)(d_l-1)+\ldots$. Distinction should be made between the cases of $g = 4k +2$ and $g =4k$. While in the former, the last layer of the tree consists of variable nodes, in the latter, it consists of check nodes. In this case, for each set of $d_l$ check nodes in the last layer of the tree, there must be at least one other variable node in $\mathcal{S} $. The sketch of the proofs for the non-elementary cases are similar to that of the elementary case, with the difference that the check node of degree $d_o$ or $d_e$ is used as the root node.
\end{proof}

\emph{Proof of Lemma} \ref{lem99}:

Consider the $d(v)$ neighbors of $v$ in $G(\mathcal{S})$. At least $d(v)-b$ of them are in $\Gamma_{\mathrm{e}}(\mathcal{S})$ and are thus connected to other variable nodes in $\mathcal{S} $. None of such variable nodes can share more than one check node from $\Gamma_{\mathrm{e}}(\mathcal{S})$ with $v$, because of the condition $g>4$. This implies that there are at least $d(v)-b$ variable nodes in $\mathcal{S}\backslash \{v\}$. \hfill  $\blacksquare$

\emph{Proof of Lemma} \ref{lem10}(i):

 Since $G(\mathcal{S})$ does not contain any cycle, it forms a tree (note that $G(\mathcal{S})$ is connected). Suppose that $G(\mathcal{S})$ is grown from a variable node of $\mathcal{S} $ as the root, one layer at a time, until along each path, the growth is terminated by reaching a check node as a leaf. These nodes are the unsatisfied check nodes of degree one. In the tree, each variable node, except the root, has a parent which is a check node of degree $\ge2$. In the case that $\mathcal{S} $ is elementary, the degree of the parent check nodes is 2, and hence each check node is the parent to one variable node. There are thus exactly $a-1$ check nodes of degree 2 in $G(\mathcal{S})$. Since $G(\mathcal{S})$ is a tree, the number of its nodes is more than the number of its edges by one. The total number of nodes in the graph is $a+(a-1)+b_1$ and the total number of edges is  $a\cdot d_v$, where $b_1$ is the number of unsatisfied check nodes of degree one. For an elementary trapping set, we thus have $2a+b_1-1=ad_v+1$, which implies $b=b_1=a(d_v-2)+2$. In the case that $\mathcal{S} $ is not elementary, some variable nodes may share the same parent. The number of parent check nodes is thus less than $a-1$, and therefore $b\ge b_1>a(d_v-2)+2$. \hfill  $\blacksquare$

\begin{lemma}
 Let $G = (L \cup R,E)$ be a left-regular bipartite graph with left degree $2$. Consider a set $\mathcal{S}\in L$, for which the induced subgraph is a tree and has the longest path of length $2k-2$. Then $|\mathcal{S}|\le \sum_{i=0}^{k-2}{(d_{c,max}-1)^{\lfloor\frac{i+1}{2} \rfloor}}$ , where $|\mathcal{S}|$ is the number of nodes in $\mathcal{S} $ and $d_{c,{max}}$ is the maximum degree of the nodes in $R$.
 \label{2chain_size}
 \end{lemma}
 \begin{proof}
The upper bound is derived by counting the number of variable nodes in a tree where the number of check nodes is maximized with the constraint that the longest path has $2k-2$ edges. This implies that there is a path of length $2k-2$ between any two leaf check nodes of $G(\mathcal{S})$. In addition, to maximize $|\mathcal{S}|$, the degree of all the check nodes in $G(\mathcal{S})$ is assumed to be $d_{c,max}$.
 \end{proof}

\section*{Acknowledgement}
The authors wish to thank the Associate Editor and the anonymous reviewers whose comments have improved the presentation of the paper.

\end{document}